\documentclass[american,aps,pra,reprint,floatfix,nofootinbib,superscriptaddress,longbibliography]{revtex4-1}
\usepackage[unicode=true,pdfusetitle, bookmarks=true,bookmarksnumbered=false,bookmarksopen=false, breaklinks=false,pdfborder={0 0 0},backref=false,colorlinks=false]{hyperref}
\hypersetup{colorlinks,linkcolor=myurlcolor,citecolor=myurlcolor,urlcolor=myurlcolor}
\usepackage{graphics,epstopdf,graphicx,amsthm,amsmath,amssymb,mathptmx,braket,colortbl,color,bm,framed,mathrsfs}
\usepackage[T1]{fontenc}
\usepackage[up]{subfigure}
\usepackage{tikz}

\definecolor{myurlcolor}{rgb}{0,0,0.9}

\newcommand{\proj}[1]{| #1\rangle\!\langle #1 |}

\DeclareMathOperator{\trace}{Tr}
\newcommand{\Ptr}[2]{\trace_{#1}\Pa{#2}}
\newcommand{\Tr}[1]{\Ptr{}{#1}}
\newcommand{\Innerm}[3]{\left\langle #1 \left| #2 \right| #3 \right\rangle}
\newcommand{\Pa}[1]{\left[#1\right]}

\newcommand{\norm}[1]{\left\lVert #1 \right\rVert}

\theoremstyle{plain}
\newtheorem{thm}{Theorem}
\newtheorem{lem}[thm]{Lemma}
\newtheorem{prop}[thm]{Proposition}
\newtheorem{cor}[thm]{Corollary}

\newcommand*{\myproofname}{Proof}
\newenvironment{mproof}[1][\myproofname]{\begin{proof}[#1]}{\end{proof}}
\def\ot{\otimes}
\def\complex{\mathbb{C}}
\def\real{\mathbb{R}}

\def\cI{\mathcal{I}}

\def\cM{\mathcal{M}}
\def\cN{\mathcal{N}}
\def\cD{\mathcal{D}}

\def\cH{\mathcal{H}}

\DeclareMathAlphabet{\mathcal}{OMS}{cmsy}{m}{n}

\makeatother

\begin{document}

  \author{Lu Li}
  \email{lilu93@zju.edu.cn}
 \affiliation{School of Mathematical Sciences, Zhejiang University, Hangzhou 310027, PR~China}
  \author{Kaifeng Bu}
 \email{kfbu@fas.harvard.edu; bkf@zju.edu.cn}
 \affiliation{School of Mathematical Sciences, Zhejiang University, Hangzhou 310027, PR~China}
\affiliation{Department of Physics, Harvard University,  Cambridge, Massachusetts 02138, USA}

  \author{Zi-Wen Liu}
 \email{zliu1@perimeterinstitute.ca; zwliu@mit.edu}
  \affiliation{Perimeter Institute for Theoretical Physics, Waterloo, Ontario N2L 2Y5, Canada}
 \affiliation{Center for Theoretical Physics, Research Laboratory of Electronics, and Department of Physics, Massachusetts Institute of Technology, Cambridge, Massachusetts 02139, USA}

\title{Quantifying the resource content of quantum channels: An operational approach}

\begin{abstract}
We propose a general method to operationally quantify the ``resourcefulness'' of quantum channels
via channel discrimination, an important information processing task.  A main result is that 
the maximum success probability of distinguishing a given channel from the 
set of free channels by free probe states is exactly characterized by the resource generating power, i.e.~the maximum amount of resource produced by the action of the channel, given by the trace distance to the set of free states.  
We apply this framework to the resource theory of quantum coherence, as an informative example.    The general results can also be easily applied to other resource theories such as entanglement, magic states, and asymmetry.

\end{abstract}

\maketitle

\section{Introduction}
Understanding and utilizing various forms of quantum resources represents a main theme of quantum information science.     To this end, a powerful framework known as the quantum resource theory is being actively developed in recent years to systematically study the quantification and manipulation of quantum resources (see \cite{ChiGourRev} for a recent review). 
In fact, the resource features of certain quantum effects, in particular quantum entanglement, have already been carefully studied earlier \cite{Plenio2007,HorodeckiRMP09,Nielsen10},
but a key observation underlying the recent interests in the resource theory framework is that the theories of different kinds of resource properties (stemming from different physical constraints) can share a largely common structure and a wide range of general approaches and results \cite{Oppenheim13,FBrandao15,Liu2017,Bartosz2017,Anshu17,Ryuji2018,PhysRevLett.122.140403,LiuBuTakagi:oneshot}.  Indeed, this idea has been successfully applied to the study of various other key quantum resources, such as 
coherence \cite{Baumgratz2014,Winter2016,RevModPhys.89.041003}, superposition \cite{Theurer2017}, magic states \cite{magic,howard_2017},
thermal non-equilibrium \cite{Fernando2013,HorodeckiOppenheim13},
asymmetry  
\cite{PhysRevA.80.012307,noether}, etc.

The well-established schemes of resource theory (at a non-abstract level; see e.g.~\cite{COECKE201659,fritz_2017} for abstract, category-theoretic formulations that do not rely on the explicit mathematical structures of the object space) mostly handle in particular static resources encoded in quantum states (density operators).    However, certain quantum processes or channels can represent dynamical quantum resources which play natural and fundamental roles in broad scenarios.  The systematic study of channel resource theories is blueprinted recently by \cite{LiuWinter2018}, but we are still at an early stage of developing the complete theory.

The quantification of resource is a central topic of all kinds of resource theories.  In particular, one is interested in the operational interpretation of certain resource measures, i.e.~how they correspond to the value of the resource in achieving some operational task. 
In state resource theories, general operational resource measures can be induced by several tasks, e.g.~resource interconversion \cite{Oppenheim13,FBrandao15,LiuBuTakagi:oneshot}, resource erasure \cite{Anshu17}. However, for quantum channels, we only know that
the smooth log-robustness characterizes the randomness cost of the task of one-shot resource erasure \cite{LiuWinter2018} at the general level.  (Note that the quantification of channel resources have been previously considered in various specific contexts, such as entanglement \cite{bip}, coherence \cite{Dana2017,Theurer2018}, non-Gaussianity \cite{ZhuangPRA18}, and magic \cite{WangWildeSu:channel_magic}).

In this work, we suggest a simple and general scheme of quantify the resourcefulness of quantum channels based on quantum channel discrimination, a fundamental problem in quantum information \cite{Acin2001, Wang2006, Pirandola18}. (Note that channel discrimination 
is already known to play key roles in the characterization of state resources \cite{Napoli2016,BuPRL2017,Ryuji2018,PhysRevLett.122.140403,PhysRevLett.122.140404}.) The core question here is how well one can distinguish a quantum channel from another by optimizing over input probe states and output measurements.
We find that the maximum success probability of distinguishing 
the given channel from the set of free operations by all free probe  states is 
exactly characterized by the maximum amount of resource that can be generated by the channel, i.e.~the resource generating power, as measured by the trace-norm distance of resource. 
This resource generating power satisfies several desirable 
properties, such 
as faithfulness, convexity,  sub-multiplicity and monotonicity.
Besides, the advantage of using a resource state as the probe state, compared with 
free probe states, is upper-bounded by the trace-norm measure of resource. 
As a prominent example, we analyze in depth the widely-studied resource theory of coherence, the structure of which allows for further results.  
Our study leads to several new understandings of the coherence theory.
This approach can be  easily generalized to many other important resource theories. As an example, we state a basic result for entanglement theory.

\section{Main results}

Given a finite dimensional Hilbert space $\cH$, let $ \mathcal{D}(\mathcal{H})$ denote the set of all quantum states 
on $\mathcal{H}$.  Assume the set of free states $\mathcal{F} $ to be a non-empty, convex and closed subset of 
$\mathcal{D}(\mathcal{H})$. 
Let $\mathfrak{F}$ be the set 
of free quantum channels, or completely positive and trace preserving (CPTP) maps.  Channels in $\mathfrak{F} $ must map all free states to free states.  

Define the
resource generating/increasing power ($\Omega/\widetilde{\Omega}$) of channel $\cN: \mathcal{D}(\mathcal{H})\longrightarrow \mathcal{D}(\mathcal{H})$ as follows.  Given some resource monotone of states $\omega$ and the set of free states $\mathcal{F}$:
\begin{eqnarray}
 \Omega(\cN) &:=& \max_{\rho\in \mathcal{F}} \omega(\cN(\rho)),\\
  \widetilde\Omega(\cN) &:=& \max_{\rho\in \mathcal{D}(\mathcal{H})} [\omega(\cN(\rho))-\omega(\rho)].
\end{eqnarray}
  Note that the complete versions  of resource generating/increasing power can also be defined which, in addition, optimize over any ancilla space (see \cite{LiuWinter2018} for extended discussions).  

A representative type of resource monotones is the distance to $\mathcal{F}$. More explicitly, given some distance measure $D$,
one can define resource measure $\omega_D$ for any quantum state $\rho$ as follows:
\begin{eqnarray}
\omega_D(\rho)
:=\min_{\sigma\in \mathcal{F}}
D(\rho, \sigma).
\end{eqnarray}
The resource generating/increasing power given by
$\omega_D$ is denoted $\Omega_D$/$\widetilde{\Omega}_D$.
It can be shown that they are actually equivalent for contractive distance metrics (see the proof in Appendix A of Supplemental Material):
 \begin{prop}\label{prop:g=i}
If the distance measure $D$ satisfies the 
triangle inequality and the data processing inequality ( i.e., non-increasing under CPTP maps), then we have
\begin{eqnarray}
\Omega_D(\cN)
=\widetilde{\Omega}_D(\cN).
\end{eqnarray}
\end{prop}
Of particular importance to this work is the trace distance $\frac{1}{2}\norm{\rho-\sigma}_1:=\frac{1}{2}\mathrm{Tr}{|\rho-\sigma|}$ , which we denote by subscript ``1''. 

Here, we aim at establishing connections between the resource generating power of a channel and its non-free feature in the task of channel discrimination.  
Given two channels $\cN$ and $\cM$,  and the same probe state $\rho$ going through the channels $\cN,\cM$ respectively,   then the success probability of distinguishing $\cN$ and $\cM$ by the probe state $\rho$ is the success probability of distinguishing 
$\cN(\rho)$ and $\cM(\rho)$ as follows
\begin{align}
  &p_{\rm succ}(\cN,\cM,\rho)\nonumber\\
   =&\max_{\{\Pi, \mathbb{I}-\Pi\}}\left\{\frac{1}{2}\Tr{\cN(\rho)\Pi}+\frac{1}{2}\Tr{\cM(\rho)(\mathbb{I}-\Pi)}\right\},~~~~
\end{align}
where the maximization is taken over all POVM $\{\Pi, \mathbb{I}-\Pi\}$.  By the Holevo-Helstrom Theorem \cite{Helstrom76},
$p_{\rm succ}(\cN,\cM,\rho)=\frac{1}{2}+\frac{1}{4}\norm{\cN(\rho)-\cM(\rho)}_1$.

The success probability of distinguishing $\cN$ from the set of channels $\mathfrak{F}$ by the probe state $\rho$ is defined as
\begin{eqnarray}
  p_{\rm succ}(\cN,\mathfrak{F},\rho):=\min_{\cM\in \mathfrak{F} }p_{\rm succ}(\cN,\cM,\rho),
\end{eqnarray}
and the maximum success probability of distinguishing $\cN$ from $\mathfrak{F}$ by using any free state or any quantum state (denoted by $Q$) as the probe state 
are respectively given by
\begin{eqnarray}
  p_{\rm succ}(\cN,\mathfrak{F}, \mathcal{F}):=\max_{\rho\in \mathcal{F} }p_{\rm succ}(\cN,\mathfrak{F},\rho), \\
  p_{\rm succ}(\cN,\mathfrak{F}, Q):=\max_{\rho\in \mathcal{D}(\cH) }p_{\rm succ}(\cN,\mathfrak{F},\rho).
\end{eqnarray}

The following result provides an exact characterization
of the success probability $ p_{\rm succ}(\cN,\mathfrak{F},\mathcal{F})$:

\begin{thm}\label{thm:1}
Given a quantum channel $\cN$ and the set of free channels $\mathfrak{F}$.  The maximum success probability of discriminating $\cN$ from $\mathfrak{F}$ by the set of free states $\mathcal{F}$ is only  directly related to the resource increasing power given by trace distance  (which equals the generating power due to Proposition \ref{prop:g=i}) of $\cN$ as follows:
\begin{equation}
    p_{\mathrm{succ}}(\cN,\mathfrak{F}, \mathcal{F}) = \frac{1}{2} + \frac{1}{2}\widetilde\Omega_{1}(\cN) = \frac{1}{2} + \frac{1}{2}\Omega_{1}(\cN).\\
\end{equation}
\end{thm}
The proof of this theorem is provided in Appendix A of Supplemental Material.
We now show that $\Omega_{1}(\cN)$ satisfies the basic conditions for resource quantifiers of quantum channels, e.g.~normalized, and monotone under left and 
right compositions with free channels \cite{LiuWinter2018}.  More specifically, 
 \begin{prop}\label{prop:pro1}
 The trace-norm resource generating power $\Omega_{1}(\cN)$
satisfies  the following properties: 
 
(i) $\Omega_{1}(\cN)\geq 0$, and 
$\Omega_{1}(\cN)=0$ if $\cN\in \mathfrak{F}$. Moreover, if $\mathfrak{F}$ includes 
all CPTP maps which maps all free states to free states (resource non-generating maps), then $\Omega_{1}(\cN)=0$ iff $\cN\in \mathfrak{F}$.

(ii) For any $\cM_1, \cM_2\in \mathfrak{F}$,  we have 
\begin{eqnarray}
 \Omega_{1}(\cM_1\circ \cN\circ \cM_2)
\leq  \Omega_{1}(\cN).
\end{eqnarray}

(iii) Given a set of quantum 
channels $\set{\cN_i, p_i}_i$ with $\sum_ip_i=1$, 
\begin{eqnarray}
\Omega_1(\sum_i p_i\cN_i)
\leq\sum_ip_i\Omega_1(\cN_i).
\end{eqnarray}
Moreover, if the free states on $\mathcal{H}_A\otimes \mathcal{H}_B$ 
is defined as convex combination of the tensor product of 
free states on $\mathcal{H}_A$ and $\mathcal{H}_B$, i.e.,
$\mathcal{F}_{AB}=Conv\set{\mathcal{F}_A\otimes \mathcal{F}_B }$, 
then resource generating power $\Omega_{1}(\cN)$ also satisfies the
following properties, 

(iv) Given two channels $\cN_1$ and $\cN_2$, it holds that 
\begin{eqnarray}
\Omega_{1}(\cN_1\ot\cN_2)\geq \max\set{\Omega_{1}(\cN_1), \Omega_{1}(\cN_2)}.
\end{eqnarray}

(v) Given two channels $\cN_1$ and $\cN_2$, it holds that 
\begin{eqnarray}
 \Omega_{1}(\cN_1\ot\cN_2) \leq \Omega_{1}(\cN_1)+ \Omega_{1}(\cN_2).
 \end{eqnarray}
 
 \end{prop}

In fact, each of the above properties holds under weaker assumptions.
The proof for more general distance measures is provided in Appendix B of Supplemental Material.
Due to property (i), Theorem \ref{thm:1} also indicates that resource non-generating channels are effectively indistinguishable from each other by free probe states.
Due to property (iv), it is easy to define a regularized version of 
$\Omega_1(\cN)$ by
$\Omega^{\infty}_1(\cN)
=\lim_{n\to\infty}\frac{1}{n}\Omega_1(\cN^{\ot n})$, which is invariant under 
tensoring, i.e., $\Omega^{\infty}_1(\cN^{\ot 2})=\Omega^{\infty}_1(\cN)$.
However, this is not the focus of this work. 

Since $\mathcal{F}\subset \mathcal{D}(\mathcal{H})$, we have 
$ p_{\rm succ}(\cN,\mathfrak{F},  Q)\geq p_{\rm succ}(\cN,\mathfrak{F}, \mathcal{F})$. If the probe state $\rho$ is not a free state, then the resource in $\rho$ may help improve the success probability of 
discriminating the given channel $\cN$ from 
the set of free channels. 
Here we provide an upper bound on the advantage of using a resource probe state:

\begin{thm}\label{thm:2}
Given a quantum channel $\cN$, a quantum state $\rho$ and  the set of free channels $\mathfrak{F}$.  
The advantage provided by the  state $\rho$ compared 
with all free states to distinguish any given channel $\cN$ from $\mathfrak{F}$ is upper bounded by the trace-norm distance of resource:
\begin{eqnarray}
p_{\mathrm{succ}}(\cN,\mathfrak{F},\rho)
-p_{\mathrm{succ}}(\cN,\mathfrak{F}, \mathcal{F}) 
\leq \frac{1}{2}\omega_1(\rho).
\end{eqnarray}

 \end{thm}

The proof is presented in Appendix C of Supplemental Material.
A direct corollary is the following bound on the success probability of discriminating $\cN$ from free channels by any probe state $\rho$:
\begin{cor}
Given a quantum channel $\cN$, a quantum state $\rho$ and  the set of free channels $\mathfrak{F}$,
 the success probability $p_{\mathrm{succ}}(\cN, \mathfrak{F}, \rho )$
is upper bounded by 
\begin{eqnarray}
p_{\mathrm{succ}}(\cN, \mathfrak{F}, \rho )
\leq \frac{1}{2}+\frac{1}{2}\Omega_1(\cN)+\frac{1}{2}\omega_1(\rho).
\end{eqnarray}
\end{cor}

\section{Example}

As an application of the above general framework, we now focus on quantum coherence, a prominent quantum feature emerging from the superposition principle of quantum mechanics. Coherence represents a key quantum resource which has a variety of applications in quantum information science, including quantum metrology \cite{Giovannetti2011}, thermodynamics \cite{Lostaglio2015,Lostaglio2015NC} and biology \cite{Plenio08, Levi14}.    In recent years, the resource theory  of coherence has drawn a lot of attention, where the manipulation and characterization of coherence in quantum states are thoroughly investigated (see \cite{RevModPhys.89.041003,Hu2018} for a review).  Now we extend the study to quantum channels following the idea in the last section, that is, to characterize the coherence value of a channel by its distinguishability from the typical sets of coherence-free channels.

Given a fixed basis $\set{\ket{i}}^{d-1}_{i=0}$ for a $d$-dimensional system, any quantum state
which is diagonal in the reference basis is called an incoherent
state and is a free state in the resource theory of coherence.
The set of incoherent states is denoted by $\cI$. Let $\Delta$ denote the fully dephasing
channel in the given basis, which is defined as
$\Delta(\rho)=\sum_i\Innerm{i}{\rho}{i}\proj{i}$.  $\Delta$ is a prominent example of the 
resource destroying map \cite{Liu2017}.

There are several individually motivated choices of free operations in the resource theory of coherence.  The following four, which collectively emerge from the relations with $\Delta$ and can be broadly generalized via the theory of resource destroying map \cite{Liu2017}, are considered most important: 
(1) maximally incoherent operations (MIO) \cite{Chitambar2016b}, the maximum possible set of coherence-free operations that contains all quantum operations
$\cM$ that maps incoherent states to incoherent states, i.e., $\cM(\cI)\subset\cI$; 
(2) incoherent operations (IO) \cite{Baumgratz2014}, containing $\cM$ that
admit a set of Kraus operators $\set{K_i}$ such that $\cM(\cdot)=\sum_iK_i(\cdot)K^\dag_i$ and
$K_i\cI K^\dag_i\subset\cI$ for any $i$;
(3) dephasing-covariant operations (DIO) \cite{Chitambar2016b,Liu2017}, containing $\cM$ such that
$[\Delta,\cM]=0$ 
(4) strictly incoherent operations (SIO) \cite{Chitambar2016a, Chitambar2016b}, containing all 
$\cM$ admitting a set of Kraus operators $\set{K_i}$ such that
$\Delta (K_i\rho K^\dag_i)=K_i\Delta(\rho) K^\dag_i$ for any $i$ and any quantum state $\rho$.

Several operational motived coherence measures have been introduced and here we consider the 
coherence measure defined by $l_1$-norm distance \cite{Baumgratz2014}, trace-norm distance \cite{Shao2015} and 
robustness \cite{Piani2016},
\begin{eqnarray}
C_{l_1}(\rho):&=&\min_{\sigma\in \cI}
\norm{\rho-\sigma}_{l_1},\\
C_1(\rho)
:&=&\frac{1}{2}\min_{\sigma\in \cI}
\norm{\rho-\sigma}_1,\\
C_R(\rho)&=&\min\set{t\geq 0 :  \rho+t\sigma\in\cI, \sigma\in \mathcal{D}(\cH)}.
\end{eqnarray}
In fact,  in single-qubit system $\complex^2$,
the trace-norm of coherence $C_1$ is equal to 
$l_1$-norm of coherence $C_{l_1}$ \cite{Rana2016pra,Shao2015} and 
the  robustness of coherence $C_R$ \cite{Piani2016} up to a scalar 2.

In the resource theory of coherence, certain coherence generating power  can also be used to characterize the cost of simulating the given channel 
by incoherent operations \cite{ Bu2017PLA,Diaz2018} and the capacity of a channel to generate maximally coherent states \cite{Dana2017}. Besides,  the ability of  a quantum channel to detect non-classicality has also been introduced to 
quantify the resource of channels in terms of trace distance \cite{Theurer2018} and relative entropy \cite{Theurer2018,Yuan2018}.

First, it follows from Theorem \ref{thm:1} that the success probability of distinguishing $\cN$ from the set of free operations $\mathfrak{I}$, where $\mathfrak{I}$ 
can be any of $\{SIO,IO,DIO,MIO\}$,  is universally determined by the trace-norm coherence generating power.

\begin{prop}\label{prop:coh1}
Given a quantum channel $\cN$ and the set of coherence-free operations $\mathfrak{I}\in\{SIO,IO,DIO,MIO\}$,
the maximum success probability of distinguishing $\cN$ from $\mathfrak{I}$ by incoherent states is
\begin{eqnarray}
  p_{\rm succ}(\cN,\mathfrak{I},\cI)  =\frac{1}{2}+\frac{1}{2}\widetilde{\mathcal{C}}_{1}(\cN)
  =\frac{1}{2}+\frac{1}{2}\mathcal{C}_{1}(\cN).
\end{eqnarray}

\end{prop}

Again, the result indicates that channels in MIO are mutually indistinguishable by incoherent states since $C_1(\cN) = 0$.
Therefore, the task of discriminating a channel from coherence-free ones gives an operational interpretation for the coherence generating power. Compared  with
 \cite{BuPRL2017} and \cite{Napoli2016}, which only consider the effect of coherence in the probe states in 
 channel discrimination, the results here reveal the roles of coherence in quantum channels in this task.

Since  the trace-norm of coherence $C_1\leq 1-1/d$ \cite{Chen2016,Yu2016}, the success probability 
$p_{\text{succ}}(\cN, \mathfrak{I}, \cI)\leq 1-1/(2d)$.  
For example, for the 
Hadamard gate $H$ on single-qubit system $\complex^2$, we have 
$p_{\rm succ}(H, \mathfrak{I}, \cI)= 3/4$, which follows from the fact that 
$\mathcal{C}_1(H)=1/2$ (see Appendix A of Supplemental Material for the calculation of  $\mathcal{C}_1$ in single-qubit
 system).
Due to the equivalence between trace-norm distance and robustness of coherence,
it may be expected that this theorem can be experimentally testified in a future work,
as the robustness of coherence can be measured  in experiment \cite{Wang2017,ZhengPRL2018}.

Obviously, 
$p_{\text{succ}}(\cN, \mathfrak{I}, Q)\geq p_{\text{succ}}(\cN, \mathfrak{I}, \cI)$ for any quantum channel. 
There exists some quantum channel $\cN$ such that 
the inequality is strict,
which shows that 
the resource of probe states is useful for distinguishing 
the given channel from the set of free operations. 

\begin{prop}\label{prop:imp}
For $\mathfrak{I}\in \set{SIO, IO}$, there exists some quantum channel $\cN$ such that 
\begin{eqnarray}
p_{\rm succ}(\cN, \mathfrak{I}, Q)>p_{\rm succ}(\cN, \mathfrak{I}, \cI).
\end{eqnarray}
\end{prop}

The proof is presented in Appendix D of Supplemental Material.
The above result shows that  resource of probe states is useful for improving the success probability of  distinguishing 
the given channel from the set of free operations $\mathfrak{I}\in \set{SIO, IO}$.
However, whether the similar result holds   for
MIO or DIO is unknown.

By applying Theorem \ref{thm:2} to the resource theory of coherence, we obtain the following upper bound 
on the success probability when we choose a coherent state as the probe state.
 
 \begin{prop}
 
Given  a set of free operations $\mathfrak{I}\in\{ SIO,IO,DIO,MIO\}$ and a probe state $\rho$. For any quantum channel $\cN$, we have
\begin{eqnarray}
p_{\rm succ}(\cN,\mathfrak{I},\rho)
-p_{\rm succ}(\cN,\mathfrak{I}, \cI)
\leq \frac{1}{2}C_1(\rho).
\end{eqnarray}

 \end{prop}

 If we restrict the measurement in the channel discrimination 
 to be an incoherent POVM, i.e., diagonal in the given basis $\set{\ket{i}}_i$, then 
  the success probability 
 to distinguish the given two channels by a
 probe state $\rho$ is 
\begin{align}
 & p^{I}_{succ}(\cN,\cM,\rho)\nonumber\\
   =&\max_{\substack{\{\Pi, \mathbb{I}-\Pi\}\\ \text{diagonal}}}\left\{\frac{1}{2}\Tr{\cN(\rho)\Pi}+\frac{1}{2}\Tr{\cM(\rho)(\mathbb{I}-\Pi)}\right\}.
\end{align}

In this case,  the success probability of distinguishing the given channel $\cN$ from the 
 set of free operation $\mathfrak{I}\in\{SIO,IO,DIO,MIO\}$
is equal to the probability of random guessing.

\begin{thm}\label{thm:coh_f}
 Given a quantum channel $\cN$ and the set of  free operations $\mathfrak{I}\in\{SIO,IO,DIO,MIO\}$, 
 then the success probability by incoherent POVM is 
 \begin{eqnarray}
 p^{I}_{\rm succ}(\cN, \mathfrak{I}, \rho)=\frac{1}{2},
 \end{eqnarray}
 for any $\rho\in \cD(\cH)$.
 \end{thm}

The proof is provided in Appendix E of Supplemental Material. Therefore, the restriction of incoherent POVM will eliminate the advantage 
 provided by the coherence of state and channel 
 in the task of  channel discrimination.
Note that Ref.~\cite{Theurer2018} considers a slightly different scenario (for example, the order of taking minimization over channels and maximization over states is different and the set of free operations there is consisted of detection-incoherent operations, which is different from those we consider), where, in contrast, it is possible to distinguish a channel from free ones with probability greater than $1/2$ even by free measurements.
 Moreover, the coherence feature of channels and its quantification that Ref.~\cite{Theurer2018} studies rely on the resource destroying map (the fully dephasing channel), but our approach does not.

 \smallskip
The general results Theorem \ref{thm:1} and \ref{thm:2} can also be applied to 
other resource theories, such as 
entanglement, magic states and so on.
For instance, in the resource theory of bipartite entanglement, 
the free states are separable states, and the free operations are
typically chosen to be Local Operations and Classical Communication (LOCC), 
or Separable operations (SEP)---the maximal set of entanglement non-generating operations. Then we have
\begin{prop}
Given  the set of free operations   $\mathfrak{I}\in\{LOCC, SEP\}$ and a probe state $\rho_{AB}\in \cD(\cH_A\ot \cH_B)$. For any quantum channel $\cN$, we have
\begin{eqnarray}
p_{\rm succ}(\cN,\mathfrak{I},\rho)
-p_{\rm succ}(\cN,\mathfrak{I}, \cI)
\leq \frac{1}{2}E_1(\rho_{AB}),
\end{eqnarray}
where $E_1(\rho_{AB}):=\min_{\sigma\in Sep(A:B)}\norm{\rho_{AB}-\sigma}_{1}$ and
$Sep(A:B) $ denotes the set of separable states on $\cH_A\ot\cH_B$.
\end{prop}

As for the free measurement case, 
in general, we  can also define the free measurement $\set{\Pi, \mathbb{I}-\Pi}$, where
 $\Pi$ and $\mathbb{I}-\Pi$ are proportional  to some free states.  If a  resource theory has   resource destroying channel $\lambda$
and $\lambda^{\dag}$ is  a resource destroying channel as well, then
  Theorem \ref{thm:coh_f} is still true (see Appendix E of Supplemental Material). 
 However, whether Theorem \ref{thm:coh_f} can be applied to other convex resource theories  is unknown.

\section{Conclusion}
This work considers the fundamental task of channel discrimination from a resource theory perspective, which leads to an intuitive and general framework of operationally quantifying the resource value of quantum channels by how efficiently they can be distinguished from the resource-free ones.
The key observation is that the maximum success probability of distinguishing a channel from the set of free operations by all free states is characterized by the trace-norm resource generating power of the channel. 
As the resource  generating power  satisfies the properties like  positivity, convexity,  sub-multiplicity and the monotonicity under free operations, it establishes an operational framework of quantifying resource in quantum channels.
We demonstrate the power of this framework in the resource theory of quantum coherence.  
In addition to the de-generalized results, we also show that restricting to incoherent POVMs in this task will eliminate any advantage over random guessing. 
Our results shed new light on the operational resource theory 
of quantum channels and in particular the resource theory of coherence.  We hope that the framework will lead to more interesting results for a variety of resource theories and information processing tasks.

\bigskip

\emph{Note added.}  During the revision of this paper, we became aware of a recent work by Liu and Yuan \cite{LiuYuan}, which establishes general connections between the resource generating/increasing power and channel distillation/dilution tasks.

\begin{acknowledgments}
This research was supported in part by the Templeton Religion Trust under grant TRT 0159. L. Li and K. Bu acknowledge Arthur Jaffe  for the support and help.
K. Bu also thanks the support of Academic Awards for Outstanding Doctoral
Candidates from Zhejiang University.  Z.-W.~Liu is supported by AFOSR, ARO, and Perimeter Institute for Theoretical Physics.  Research at Perimeter Institute is supported by the Government
of Canada through Industry Canada and by the Province of Ontario through the Ministry
of Research and Innovation.

\end{acknowledgments}

 \bibliography{Maxcoh-lit}

\begin{thebibliography}{61}%
\makeatletter
\providecommand \@ifxundefined [1]{%
 \@ifx{#1\undefined}
}%
\providecommand \@ifnum [1]{%
 \ifnum #1\expandafter \@firstoftwo
 \else \expandafter \@secondoftwo
 \fi
}%
\providecommand \@ifx [1]{%
 \ifx #1\expandafter \@firstoftwo
 \else \expandafter \@secondoftwo
 \fi
}%
\providecommand \natexlab [1]{#1}%
\providecommand \enquote  [1]{``#1''}%
\providecommand \bibnamefont  [1]{#1}%
\providecommand \bibfnamefont [1]{#1}%
\providecommand \citenamefont [1]{#1}%
\providecommand \href@noop [0]{\@secondoftwo}%
\providecommand \href [0]{\begingroup \@sanitize@url \@href}%
\providecommand \@href[1]{\@@startlink{#1}\@@href}%
\providecommand \@@href[1]{\endgroup#1\@@endlink}%
\providecommand \@sanitize@url [0]{\catcode `\\12\catcode `\$12\catcode
  `\&12\catcode `\#12\catcode `\^12\catcode `\_12\catcode `\%12\relax}%
\providecommand \@@startlink[1]{}%
\providecommand \@@endlink[0]{}%
\providecommand \url  [0]{\begingroup\@sanitize@url \@url }%
\providecommand \@url [1]{\endgroup\@href {#1}{\urlprefix }}%
\providecommand \urlprefix  [0]{URL }%
\providecommand \Eprint [0]{\href }%
\providecommand \doibase [0]{http://dx.doi.org/}%
\providecommand \selectlanguage [0]{\@gobble}%
\providecommand \bibinfo  [0]{\@secondoftwo}%
\providecommand \bibfield  [0]{\@secondoftwo}%
\providecommand \translation [1]{[#1]}%
\providecommand \BibitemOpen [0]{}%
\providecommand \bibitemStop [0]{}%
\providecommand \bibitemNoStop [0]{.\EOS\space}%
\providecommand \EOS [0]{\spacefactor3000\relax}%
\providecommand \BibitemShut  [1]{\csname bibitem#1\endcsname}%
\let\auto@bib@innerbib\@empty
\bibitem [{\citenamefont {{Chitambar}}\ and\ \citenamefont
  {{Gour}}()}]{ChiGourRev}%
  \BibitemOpen
  \bibfield  {author} {\bibinfo {author} {\bibfnamefont {Eric}\ \bibnamefont
  {{Chitambar}}}\ and\ \bibinfo {author} {\bibfnamefont {Gilad}\ \bibnamefont
  {{Gour}}},\ }\bibfield  {title} {\enquote {\bibinfo {title} {{Quantum
  Resource Theories}},}\ }\href@noop {} {\ }\Eprint
  {http://arxiv.org/abs/1806.06107} {arXiv:1806.06107} \BibitemShut {NoStop}%
\bibitem [{\citenamefont {Plenio}\ and\ \citenamefont
  {Virmani}(2007)}]{Plenio2007}%
  \BibitemOpen
  \bibfield  {author} {\bibinfo {author} {\bibfnamefont {Martin~B}\
  \bibnamefont {Plenio}}\ and\ \bibinfo {author} {\bibfnamefont {Shashank}\
  \bibnamefont {Virmani}},\ }\bibfield  {title} {\enquote {\bibinfo {title} {An
  introduction to entanglement measures},}\ }\href@noop {} {\bibfield
  {journal} {\bibinfo  {journal} {Quan. Info. Comput.}\ }\textbf {\bibinfo
  {volume} {7}},\ \bibinfo {pages} {001--051} (\bibinfo {year}
  {2007})}\BibitemShut {NoStop}%
\bibitem [{\citenamefont {Horodecki}\ \emph {et~al.}(2009)\citenamefont
  {Horodecki}, \citenamefont {Horodecki}, \citenamefont {Horodecki},\ and\
  \citenamefont {Horodecki}}]{HorodeckiRMP09}%
  \BibitemOpen
  \bibfield  {author} {\bibinfo {author} {\bibfnamefont {Ryszard}\ \bibnamefont
  {Horodecki}}, \bibinfo {author} {\bibfnamefont {Pawe\l{}}\ \bibnamefont
  {Horodecki}}, \bibinfo {author} {\bibfnamefont {Micha\l{}}\ \bibnamefont
  {Horodecki}}, \ and\ \bibinfo {author} {\bibfnamefont {Karol}\ \bibnamefont
  {Horodecki}},\ }\bibfield  {title} {\enquote {\bibinfo {title} {Quantum
  entanglement},}\ }\href {\doibase 10.1103/RevModPhys.81.865} {\bibfield
  {journal} {\bibinfo  {journal} {Rev. Mod. Phys.}\ }\textbf {\bibinfo {volume}
  {81}},\ \bibinfo {pages} {865--942} (\bibinfo {year} {2009})}\BibitemShut
  {NoStop}%
\bibitem [{\citenamefont {Nielsen}\ and\ \citenamefont
  {Chuang}(2010)}]{Nielsen10}%
  \BibitemOpen
  \bibfield  {author} {\bibinfo {author} {\bibfnamefont {Michael~A.}\
  \bibnamefont {Nielsen}}\ and\ \bibinfo {author} {\bibfnamefont {Isaac~L.}\
  \bibnamefont {Chuang}},\ }\href {http://dx.doi.org/10.1017/CBO9780511976667}
  {\emph {\bibinfo {title} {Quantum Computation and Quantum Information}}}\
  (\bibinfo  {publisher} {Cambridge University Press},\ \bibinfo {year}
  {2010})\BibitemShut {NoStop}%
\bibitem [{\citenamefont {Horodecki}\ and\ \citenamefont
  {Oppenheim}(2013{\natexlab{a}})}]{Oppenheim13}%
  \BibitemOpen
  \bibfield  {author} {\bibinfo {author} {\bibfnamefont {Michal}\ \bibnamefont
  {Horodecki}}\ and\ \bibinfo {author} {\bibfnamefont {Jonathan}\ \bibnamefont
  {Oppenheim}},\ }\bibfield  {title} {\enquote {\bibinfo {title} {(quantumness
  in the context of) resource theories},}\ }\href {\doibase
  10.1142/S0217979213450197} {\bibfield  {journal} {\bibinfo  {journal} {Int.
  J. Mod. Phys. B}\ }\textbf {\bibinfo {volume} {27}},\ \bibinfo {pages}
  {1345019} (\bibinfo {year} {2013}{\natexlab{a}})}\BibitemShut {NoStop}%
\bibitem [{\citenamefont {Brand\~ao}\ and\ \citenamefont
  {Gour}(2015)}]{FBrandao15}%
  \BibitemOpen
  \bibfield  {author} {\bibinfo {author} {\bibfnamefont {Fernando G. S.~L.}\
  \bibnamefont {Brand\~ao}}\ and\ \bibinfo {author} {\bibfnamefont {Gilad}\
  \bibnamefont {Gour}},\ }\bibfield  {title} {\enquote {\bibinfo {title}
  {Reversible framework for quantum resource theories},}\ }\href {\doibase
  10.1103/PhysRevLett.115.070503} {\bibfield  {journal} {\bibinfo  {journal}
  {Phys. Rev. Lett.}\ }\textbf {\bibinfo {volume} {115}},\ \bibinfo {pages}
  {070503} (\bibinfo {year} {2015})}\BibitemShut {NoStop}%
\bibitem [{\citenamefont {Liu}\ \emph {et~al.}(2017)\citenamefont {Liu},
  \citenamefont {Hu},\ and\ \citenamefont {Lloyd}}]{Liu2017}%
  \BibitemOpen
  \bibfield  {author} {\bibinfo {author} {\bibfnamefont {Zi-Wen}\ \bibnamefont
  {Liu}}, \bibinfo {author} {\bibfnamefont {Xueyuan}\ \bibnamefont {Hu}}, \
  and\ \bibinfo {author} {\bibfnamefont {Seth}\ \bibnamefont {Lloyd}},\
  }\bibfield  {title} {\enquote {\bibinfo {title} {Resource destroying maps},}\
  }\href {\doibase 10.1103/PhysRevLett.118.060502} {\bibfield  {journal}
  {\bibinfo  {journal} {Phys. Rev. Lett.}\ }\textbf {\bibinfo {volume} {118}},\
  \bibinfo {pages} {060502} (\bibinfo {year} {2017})}\BibitemShut {NoStop}%
\bibitem [{\citenamefont {Regula}(2018)}]{Bartosz2017}%
  \BibitemOpen
  \bibfield  {author} {\bibinfo {author} {\bibfnamefont {Bartosz}\ \bibnamefont
  {Regula}},\ }\bibfield  {title} {\enquote {\bibinfo {title} {Convex geometry
  of quantum resource quantification},}\ }\href
  {http://stacks.iop.org/1751-8121/51/i=4/a=045303} {\bibfield  {journal}
  {\bibinfo  {journal} {J. Phys. A}\ }\textbf {\bibinfo {volume} {51}},\
  \bibinfo {pages} {045303} (\bibinfo {year} {2018})}\BibitemShut {NoStop}%
\bibitem [{\citenamefont {Anshu}\ \emph {et~al.}(2018)\citenamefont {Anshu},
  \citenamefont {Hsieh},\ and\ \citenamefont {Jain}}]{Anshu17}%
  \BibitemOpen
  \bibfield  {author} {\bibinfo {author} {\bibfnamefont {Anurag}\ \bibnamefont
  {Anshu}}, \bibinfo {author} {\bibfnamefont {Min-Hsiu}\ \bibnamefont {Hsieh}},
  \ and\ \bibinfo {author} {\bibfnamefont {Rahul}\ \bibnamefont {Jain}},\
  }\bibfield  {title} {\enquote {\bibinfo {title} {Quantifying resources in
  general resource theory with catalysts},}\ }\href {\doibase
  10.1103/PhysRevLett.121.190504} {\bibfield  {journal} {\bibinfo  {journal}
  {Phys. Rev. Lett.}\ }\textbf {\bibinfo {volume} {121}},\ \bibinfo {pages}
  {190504} (\bibinfo {year} {2018})}\BibitemShut {NoStop}%
\bibitem [{\citenamefont {Takagi}\ \emph {et~al.}(2019)\citenamefont {Takagi},
  \citenamefont {Regula}, \citenamefont {Bu}, \citenamefont {Liu},\ and\
  \citenamefont {Adesso}}]{Ryuji2018}%
  \BibitemOpen
  \bibfield  {author} {\bibinfo {author} {\bibfnamefont {Ryuji}\ \bibnamefont
  {Takagi}}, \bibinfo {author} {\bibfnamefont {Bartosz}\ \bibnamefont
  {Regula}}, \bibinfo {author} {\bibfnamefont {Kaifeng}\ \bibnamefont {Bu}},
  \bibinfo {author} {\bibfnamefont {Zi-Wen}\ \bibnamefont {Liu}}, \ and\
  \bibinfo {author} {\bibfnamefont {Gerardo}\ \bibnamefont {Adesso}},\
  }\bibfield  {title} {\enquote {\bibinfo {title} {Operational advantage of
  quantum resources in subchannel discrimination},}\ }\href {\doibase
  10.1103/PhysRevLett.122.140402} {\bibfield  {journal} {\bibinfo  {journal}
  {Phys. Rev. Lett.}\ }\textbf {\bibinfo {volume} {122}},\ \bibinfo {pages}
  {140402} (\bibinfo {year} {2019})}\BibitemShut {NoStop}%
\bibitem [{\citenamefont {Skrzypczyk}\ and\ \citenamefont
  {Linden}(2019)}]{PhysRevLett.122.140403}%
  \BibitemOpen
  \bibfield  {author} {\bibinfo {author} {\bibfnamefont {Paul}\ \bibnamefont
  {Skrzypczyk}}\ and\ \bibinfo {author} {\bibfnamefont {Noah}\ \bibnamefont
  {Linden}},\ }\bibfield  {title} {\enquote {\bibinfo {title} {Robustness of
  measurement, discrimination games, and accessible information},}\ }\href
  {\doibase 10.1103/PhysRevLett.122.140403} {\bibfield  {journal} {\bibinfo
  {journal} {Phys. Rev. Lett.}\ }\textbf {\bibinfo {volume} {122}},\ \bibinfo
  {pages} {140403} (\bibinfo {year} {2019})}\BibitemShut {NoStop}%
\bibitem [{\citenamefont {{Liu}}\ \emph {et~al.}(2019)\citenamefont {{Liu}},
  \citenamefont {{Bu}},\ and\ \citenamefont {{Takagi}}}]{LiuBuTakagi:oneshot}%
  \BibitemOpen
  \bibfield  {author} {\bibinfo {author} {\bibfnamefont {Zi-Wen}\ \bibnamefont
  {{Liu}}}, \bibinfo {author} {\bibfnamefont {Kaifeng}\ \bibnamefont {{Bu}}}, \
  and\ \bibinfo {author} {\bibfnamefont {Ryuji}\ \bibnamefont {{Takagi}}},\
  }\bibfield  {title} {\enquote {\bibinfo {title} {{One-shot operational
  quantum resource theory}},}\ }\href@noop {} {\bibfield  {journal} {\bibinfo
  {journal} {arXiv e-prints}\ ,\ \bibinfo {eid} {arXiv:1904.05840}} (\bibinfo
  {year} {2019})},\ \Eprint {http://arxiv.org/abs/1904.05840} {arXiv:1904.05840
  [quant-ph]} \BibitemShut {NoStop}%
\bibitem [{\citenamefont {Baumgratz}\ \emph {et~al.}(2014)\citenamefont
  {Baumgratz}, \citenamefont {Cramer},\ and\ \citenamefont
  {Plenio}}]{Baumgratz2014}%
  \BibitemOpen
  \bibfield  {author} {\bibinfo {author} {\bibfnamefont {T.}~\bibnamefont
  {Baumgratz}}, \bibinfo {author} {\bibfnamefont {M.}~\bibnamefont {Cramer}}, \
  and\ \bibinfo {author} {\bibfnamefont {M.~B.}\ \bibnamefont {Plenio}},\
  }\bibfield  {title} {\enquote {\bibinfo {title} {Quantifying coherence},}\
  }\href {\doibase 10.1103/PhysRevLett.113.140401} {\bibfield  {journal}
  {\bibinfo  {journal} {Phys. Rev. Lett.}\ }\textbf {\bibinfo {volume} {113}},\
  \bibinfo {pages} {140401} (\bibinfo {year} {2014})}\BibitemShut {NoStop}%
\bibitem [{\citenamefont {Winter}\ and\ \citenamefont
  {Yang}(2016)}]{Winter2016}%
  \BibitemOpen
  \bibfield  {author} {\bibinfo {author} {\bibfnamefont {Andreas}\ \bibnamefont
  {Winter}}\ and\ \bibinfo {author} {\bibfnamefont {Dong}\ \bibnamefont
  {Yang}},\ }\bibfield  {title} {\enquote {\bibinfo {title} {Operational
  resource theory of coherence},}\ }\href {\doibase
  10.1103/PhysRevLett.116.120404} {\bibfield  {journal} {\bibinfo  {journal}
  {Phys. Rev. Lett.}\ }\textbf {\bibinfo {volume} {116}},\ \bibinfo {pages}
  {120404} (\bibinfo {year} {2016})}\BibitemShut {NoStop}%
\bibitem [{\citenamefont {Streltsov}\ \emph {et~al.}(2017)\citenamefont
  {Streltsov}, \citenamefont {Adesso},\ and\ \citenamefont
  {Plenio}}]{RevModPhys.89.041003}%
  \BibitemOpen
  \bibfield  {author} {\bibinfo {author} {\bibfnamefont {Alexander}\
  \bibnamefont {Streltsov}}, \bibinfo {author} {\bibfnamefont {Gerardo}\
  \bibnamefont {Adesso}}, \ and\ \bibinfo {author} {\bibfnamefont {Martin~B.}\
  \bibnamefont {Plenio}},\ }\bibfield  {title} {\enquote {\bibinfo {title}
  {Colloquium: Coherence},}\ }\href {\doibase 10.1103/RevModPhys.89.041003}
  {\bibfield  {journal} {\bibinfo  {journal} {Rev. Mod. Phys.}\ }\textbf
  {\bibinfo {volume} {89}},\ \bibinfo {pages} {041003} (\bibinfo {year}
  {2017})}\BibitemShut {NoStop}%
\bibitem [{\citenamefont {Theurer}\ \emph {et~al.}(2017)\citenamefont
  {Theurer}, \citenamefont {Killoran}, \citenamefont {Egloff},\ and\
  \citenamefont {Plenio}}]{Theurer2017}%
  \BibitemOpen
  \bibfield  {author} {\bibinfo {author} {\bibfnamefont {T.}~\bibnamefont
  {Theurer}}, \bibinfo {author} {\bibfnamefont {N.}~\bibnamefont {Killoran}},
  \bibinfo {author} {\bibfnamefont {D.}~\bibnamefont {Egloff}}, \ and\ \bibinfo
  {author} {\bibfnamefont {M.~B.}\ \bibnamefont {Plenio}},\ }\bibfield  {title}
  {\enquote {\bibinfo {title} {Resource theory of superposition},}\ }\href
  {\doibase 10.1103/PhysRevLett.119.230401} {\bibfield  {journal} {\bibinfo
  {journal} {Phys. Rev. Lett.}\ }\textbf {\bibinfo {volume} {119}},\ \bibinfo
  {pages} {230401} (\bibinfo {year} {2017})}\BibitemShut {NoStop}%
\bibitem [{\citenamefont {Veitch}\ \emph {et~al.}(2014)\citenamefont {Veitch},
  \citenamefont {Mousavian}, \citenamefont {Gottesman},\ and\ \citenamefont
  {Emerson}}]{magic}%
  \BibitemOpen
  \bibfield  {author} {\bibinfo {author} {\bibfnamefont {Victor}\ \bibnamefont
  {Veitch}}, \bibinfo {author} {\bibfnamefont {S~A~Hamed}\ \bibnamefont
  {Mousavian}}, \bibinfo {author} {\bibfnamefont {Daniel}\ \bibnamefont
  {Gottesman}}, \ and\ \bibinfo {author} {\bibfnamefont {Joseph}\ \bibnamefont
  {Emerson}},\ }\bibfield  {title} {\enquote {\bibinfo {title} {The resource
  theory of stabilizer quantum computation},}\ }\href
  {http://stacks.iop.org/1367-2630/16/i=1/a=013009} {\bibfield  {journal}
  {\bibinfo  {journal} {New J. Phys.}\ }\textbf {\bibinfo {volume} {16}},\
  \bibinfo {pages} {013009} (\bibinfo {year} {2014})}\BibitemShut {NoStop}%
\bibitem [{\citenamefont {Howard}\ and\ \citenamefont
  {Campbell}(2017)}]{howard_2017}%
  \BibitemOpen
  \bibfield  {author} {\bibinfo {author} {\bibfnamefont {Mark}\ \bibnamefont
  {Howard}}\ and\ \bibinfo {author} {\bibfnamefont {Earl}\ \bibnamefont
  {Campbell}},\ }\bibfield  {title} {\enquote {\bibinfo {title} {Application of
  a {{Resource Theory}} for {{Magic States}} to {{Fault}}-{{Tolerant Quantum
  Computing}}},}\ }\href {\doibase 10.1103/PhysRevLett.118.090501} {\bibfield
  {journal} {\bibinfo  {journal} {Phys. Rev. Lett.}\ }\textbf {\bibinfo
  {volume} {118}},\ \bibinfo {pages} {090501} (\bibinfo {year}
  {2017})}\BibitemShut {NoStop}%
\bibitem [{\citenamefont {Brand\~ao}\ \emph {et~al.}(2013)\citenamefont
  {Brand\~ao}, \citenamefont {Horodecki}, \citenamefont {Oppenheim},
  \citenamefont {Renes},\ and\ \citenamefont {Spekkens}}]{Fernando2013}%
  \BibitemOpen
  \bibfield  {author} {\bibinfo {author} {\bibfnamefont {Fernando G. S.~L.}\
  \bibnamefont {Brand\~ao}}, \bibinfo {author} {\bibfnamefont {Micha\l{}}\
  \bibnamefont {Horodecki}}, \bibinfo {author} {\bibfnamefont {Jonathan}\
  \bibnamefont {Oppenheim}}, \bibinfo {author} {\bibfnamefont {Joseph~M.}\
  \bibnamefont {Renes}}, \ and\ \bibinfo {author} {\bibfnamefont {Robert~W.}\
  \bibnamefont {Spekkens}},\ }\bibfield  {title} {\enquote {\bibinfo {title}
  {Resource theory of quantum states out of thermal equilibrium},}\ }\href
  {\doibase 10.1103/PhysRevLett.111.250404} {\bibfield  {journal} {\bibinfo
  {journal} {Phys. Rev. Lett.}\ }\textbf {\bibinfo {volume} {111}},\ \bibinfo
  {pages} {250404} (\bibinfo {year} {2013})}\BibitemShut {NoStop}%
\bibitem [{\citenamefont {Horodecki}\ and\ \citenamefont
  {Oppenheim}(2013{\natexlab{b}})}]{HorodeckiOppenheim13}%
  \BibitemOpen
  \bibfield  {author} {\bibinfo {author} {\bibfnamefont {Micha{\l}}\
  \bibnamefont {Horodecki}}\ and\ \bibinfo {author} {\bibfnamefont {Jonathan}\
  \bibnamefont {Oppenheim}},\ }\bibfield  {title} {\enquote {\bibinfo {title}
  {Fundamental limitations for quantum and nanoscale thermodynamics},}\ }\href
  {\doibase 10.1038/ncomms3059} {\bibfield  {journal} {\bibinfo  {journal}
  {Nat. Commun.}\ }\textbf {\bibinfo {volume} {4}},\ \bibinfo {pages} {2059}
  (\bibinfo {year} {2013}{\natexlab{b}})}\BibitemShut {NoStop}%
\bibitem [{\citenamefont {Gour}\ \emph {et~al.}(2009)\citenamefont {Gour},
  \citenamefont {Marvian},\ and\ \citenamefont
  {Spekkens}}]{PhysRevA.80.012307}%
  \BibitemOpen
  \bibfield  {author} {\bibinfo {author} {\bibfnamefont {Gilad}\ \bibnamefont
  {Gour}}, \bibinfo {author} {\bibfnamefont {Iman}\ \bibnamefont {Marvian}}, \
  and\ \bibinfo {author} {\bibfnamefont {Robert~W.}\ \bibnamefont {Spekkens}},\
  }\bibfield  {title} {\enquote {\bibinfo {title} {Measuring the quality of a
  quantum reference frame: The relative entropy of frameness},}\ }\href
  {\doibase 10.1103/PhysRevA.80.012307} {\bibfield  {journal} {\bibinfo
  {journal} {Phys. Rev. A}\ }\textbf {\bibinfo {volume} {80}},\ \bibinfo
  {pages} {012307} (\bibinfo {year} {2009})}\BibitemShut {NoStop}%
\bibitem [{\citenamefont {Marvian}\ and\ \citenamefont
  {Spekkens}(2014)}]{noether}%
  \BibitemOpen
  \bibfield  {author} {\bibinfo {author} {\bibfnamefont {Iman}\ \bibnamefont
  {Marvian}}\ and\ \bibinfo {author} {\bibfnamefont {Robert~W.}\ \bibnamefont
  {Spekkens}},\ }\bibfield  {title} {\enquote {\bibinfo {title} {Extending
  noether's theorem by quantifying the asymmetry of quantum states},}\ }\href
  {\doibase 10.1038/ncomms4821} {\bibfield  {journal} {\bibinfo  {journal}
  {Nat. Commun.}\ }\textbf {\bibinfo {volume} {5}},\ \bibinfo {pages} {3821}
  (\bibinfo {year} {2014})}\BibitemShut {NoStop}%
\bibitem [{\citenamefont {Coecke}\ \emph {et~al.}(2016)\citenamefont {Coecke},
  \citenamefont {Fritz},\ and\ \citenamefont {Spekkens}}]{COECKE201659}%
  \BibitemOpen
  \bibfield  {author} {\bibinfo {author} {\bibfnamefont {Bob}\ \bibnamefont
  {Coecke}}, \bibinfo {author} {\bibfnamefont {Tobias}\ \bibnamefont {Fritz}},
  \ and\ \bibinfo {author} {\bibfnamefont {Robert~W.}\ \bibnamefont
  {Spekkens}},\ }\bibfield  {title} {\enquote {\bibinfo {title} {A mathematical
  theory of resources},}\ }\href {\doibase 10.1016/j.ic.2016.02.008} {\bibfield
   {journal} {\bibinfo  {journal} {Inf. Comput.}\ }\textbf {\bibinfo {volume}
  {250}},\ \bibinfo {pages} {59--86} (\bibinfo {year} {2016})}\BibitemShut
  {NoStop}%
\bibitem [{\citenamefont {Fritz}(2017)}]{fritz_2017}%
  \BibitemOpen
  \bibfield  {author} {\bibinfo {author} {\bibfnamefont {Tobias}\ \bibnamefont
  {Fritz}},\ }\bibfield  {title} {\enquote {\bibinfo {title} {Resource
  convertibility and ordered commutative monoids},}\ }\href {\doibase
  10.1017/S0960129515000444} {\bibfield  {journal} {\bibinfo  {journal} {Math.
  Struct. Comput. Sci.}\ }\textbf {\bibinfo {volume} {27}},\ \bibinfo {pages}
  {850--938} (\bibinfo {year} {2017})}\BibitemShut {NoStop}%
\bibitem [{\citenamefont {{Liu}}\ and\ \citenamefont
  {{Winter}}(2019)}]{LiuWinter2018}%
  \BibitemOpen
  \bibfield  {author} {\bibinfo {author} {\bibfnamefont {Zi-Wen}\ \bibnamefont
  {{Liu}}}\ and\ \bibinfo {author} {\bibfnamefont {Andreas}\ \bibnamefont
  {{Winter}}},\ }\bibfield  {title} {\enquote {\bibinfo {title} {{Resource
  theories of quantum channels and the universal role of resource erasure}},}\
  }\href@noop {} {\bibfield  {journal} {\bibinfo  {journal} {arXiv e-prints}\
  ,\ \bibinfo {eid} {arXiv:1904.04201}} (\bibinfo {year} {2019})},\ \Eprint
  {http://arxiv.org/abs/1904.04201} {arXiv:1904.04201 [quant-ph]} \BibitemShut
  {NoStop}%
\bibitem [{\citenamefont {Bennett}\ \emph {et~al.}(2003)\citenamefont
  {Bennett}, \citenamefont {Harrow}, \citenamefont {Leung},\ and\ \citenamefont
  {Smolin}}]{bip}%
  \BibitemOpen
  \bibfield  {author} {\bibinfo {author} {\bibfnamefont {Charles~H.}\
  \bibnamefont {Bennett}}, \bibinfo {author} {\bibfnamefont {Aram~W.}\
  \bibnamefont {Harrow}}, \bibinfo {author} {\bibfnamefont {Debbie~W.}\
  \bibnamefont {Leung}}, \ and\ \bibinfo {author} {\bibfnamefont {John~A.}\
  \bibnamefont {Smolin}},\ }\bibfield  {title} {\enquote {\bibinfo {title} {On
  the capacities of bipartite hamiltonians and unitary gates},}\ }\href
  {\doibase 10.1109/TIT.2003.814935} {\bibfield  {journal} {\bibinfo  {journal}
  {IEEE Transactions on Information Theory}\ }\textbf {\bibinfo {volume}
  {49}},\ \bibinfo {pages} {1895--1911} (\bibinfo {year} {2003})}\BibitemShut
  {NoStop}%
\bibitem [{\citenamefont {Ben~Dana}\ \emph {et~al.}(2017)\citenamefont
  {Ben~Dana}, \citenamefont {Garc\'{\i}a~D\'{\i}az}, \citenamefont {Mejatty},\
  and\ \citenamefont {Winter}}]{Dana2017}%
  \BibitemOpen
  \bibfield  {author} {\bibinfo {author} {\bibfnamefont {Khaled}\ \bibnamefont
  {Ben~Dana}}, \bibinfo {author} {\bibfnamefont {Mar\'{\i}a}\ \bibnamefont
  {Garc\'{\i}a~D\'{\i}az}}, \bibinfo {author} {\bibfnamefont {Mohamed}\
  \bibnamefont {Mejatty}}, \ and\ \bibinfo {author} {\bibfnamefont {Andreas}\
  \bibnamefont {Winter}},\ }\bibfield  {title} {\enquote {\bibinfo {title}
  {Resource theory of coherence: Beyond states},}\ }\href {\doibase
  10.1103/PhysRevA.95.062327} {\bibfield  {journal} {\bibinfo  {journal} {Phys.
  Rev. A}\ }\textbf {\bibinfo {volume} {95}},\ \bibinfo {pages} {062327}
  (\bibinfo {year} {2017})}\BibitemShut {NoStop}%
\bibitem [{\citenamefont {Theurer}\ \emph {et~al.}()\citenamefont {Theurer},
  \citenamefont {Egloff}, \citenamefont {Zhang},\ and\ \citenamefont
  {Plenio}}]{Theurer2018}%
  \BibitemOpen
  \bibfield  {author} {\bibinfo {author} {\bibfnamefont {Thomas}\ \bibnamefont
  {Theurer}}, \bibinfo {author} {\bibfnamefont {Dario}\ \bibnamefont {Egloff}},
  \bibinfo {author} {\bibfnamefont {Lijian}\ \bibnamefont {Zhang}}, \ and\
  \bibinfo {author} {\bibfnamefont {Martin~B.}\ \bibnamefont {Plenio}},\
  }\bibfield  {title} {\enquote {\bibinfo {title} {Quantifying the coherence of
  operations},}\ }\href@noop {} {\ }\Eprint {http://arxiv.org/abs/1806.07332}
  {arXiv:1806.07332} \BibitemShut {NoStop}%
\bibitem [{\citenamefont {Zhuang}\ \emph {et~al.}(2018)\citenamefont {Zhuang},
  \citenamefont {Shor},\ and\ \citenamefont {Shapiro}}]{ZhuangPRA18}%
  \BibitemOpen
  \bibfield  {author} {\bibinfo {author} {\bibfnamefont {Quntao}\ \bibnamefont
  {Zhuang}}, \bibinfo {author} {\bibfnamefont {Peter~W.}\ \bibnamefont {Shor}},
  \ and\ \bibinfo {author} {\bibfnamefont {Jeffrey~H.}\ \bibnamefont
  {Shapiro}},\ }\bibfield  {title} {\enquote {\bibinfo {title} {Resource theory
  of non-gaussian operations},}\ }\href {\doibase 10.1103/PhysRevA.97.052317}
  {\bibfield  {journal} {\bibinfo  {journal} {Phys. Rev. A}\ }\textbf {\bibinfo
  {volume} {97}},\ \bibinfo {pages} {052317} (\bibinfo {year}
  {2018})}\BibitemShut {NoStop}%
\bibitem [{\citenamefont {{Wang}}\ \emph {et~al.}(2019)\citenamefont {{Wang}},
  \citenamefont {{Wilde}},\ and\ \citenamefont
  {{Su}}}]{WangWildeSu:channel_magic}%
  \BibitemOpen
  \bibfield  {author} {\bibinfo {author} {\bibfnamefont {Xin}\ \bibnamefont
  {{Wang}}}, \bibinfo {author} {\bibfnamefont {Mark~M.}\ \bibnamefont
  {{Wilde}}}, \ and\ \bibinfo {author} {\bibfnamefont {Yuan}\ \bibnamefont
  {{Su}}},\ }\bibfield  {title} {\enquote {\bibinfo {title} {{Quantifying the
  magic of quantum channels}},}\ }\href@noop {} {\bibfield  {journal} {\bibinfo
   {journal} {arXiv e-prints}\ ,\ \bibinfo {eid} {arXiv:1903.04483}} (\bibinfo
  {year} {2019})},\ \Eprint {http://arxiv.org/abs/1903.04483} {arXiv:1903.04483
  [quant-ph]} \BibitemShut {NoStop}%
\bibitem [{\citenamefont {Ac\'{\i}n}(2001)}]{Acin2001}%
  \BibitemOpen
  \bibfield  {author} {\bibinfo {author} {\bibfnamefont {A.}~\bibnamefont
  {Ac\'{\i}n}},\ }\bibfield  {title} {\enquote {\bibinfo {title} {Statistical
  distinguishability between unitary operations},}\ }\href {\doibase
  10.1103/PhysRevLett.87.177901} {\bibfield  {journal} {\bibinfo  {journal}
  {Phys. Rev. Lett.}\ }\textbf {\bibinfo {volume} {87}},\ \bibinfo {pages}
  {177901} (\bibinfo {year} {2001})}\BibitemShut {NoStop}%
\bibitem [{\citenamefont {Wang}\ and\ \citenamefont {Ying}(2006)}]{Wang2006}%
  \BibitemOpen
  \bibfield  {author} {\bibinfo {author} {\bibfnamefont {Guoming}\ \bibnamefont
  {Wang}}\ and\ \bibinfo {author} {\bibfnamefont {Mingsheng}\ \bibnamefont
  {Ying}},\ }\bibfield  {title} {\enquote {\bibinfo {title} {Unambiguous
  discrimination among quantum operations},}\ }\href {\doibase
  10.1103/PhysRevA.73.042301} {\bibfield  {journal} {\bibinfo  {journal} {Phys.
  Rev. A}\ }\textbf {\bibinfo {volume} {73}},\ \bibinfo {pages} {042301}
  (\bibinfo {year} {2006})}\BibitemShut {NoStop}%
\bibitem [{\citenamefont {{Pirandola}}\ \emph {et~al.}()\citenamefont
  {{Pirandola}}, \citenamefont {{Laurenza}},\ and\ \citenamefont
  {{Lupo}}}]{Pirandola18}%
  \BibitemOpen
  \bibfield  {author} {\bibinfo {author} {\bibfnamefont {Stefano}\ \bibnamefont
  {{Pirandola}}}, \bibinfo {author} {\bibfnamefont {Riccardo}\ \bibnamefont
  {{Laurenza}}}, \ and\ \bibinfo {author} {\bibfnamefont {Cosmo}\ \bibnamefont
  {{Lupo}}},\ }\bibfield  {title} {\enquote {\bibinfo {title} {{Fundamental
  limits to quantum channel discrimination}},}\ }\href@noop {} {\ }\Eprint
  {http://arxiv.org/abs/1803.02834} {arXiv:1803.02834} \BibitemShut {NoStop}%
\bibitem [{\citenamefont {Napoli}\ \emph {et~al.}(2016)\citenamefont {Napoli},
  \citenamefont {Bromley}, \citenamefont {Cianciaruso}, \citenamefont {Piani},
  \citenamefont {Johnston},\ and\ \citenamefont {Adesso}}]{Napoli2016}%
  \BibitemOpen
  \bibfield  {author} {\bibinfo {author} {\bibfnamefont {Carmine}\ \bibnamefont
  {Napoli}}, \bibinfo {author} {\bibfnamefont {Thomas~R.}\ \bibnamefont
  {Bromley}}, \bibinfo {author} {\bibfnamefont {Marco}\ \bibnamefont
  {Cianciaruso}}, \bibinfo {author} {\bibfnamefont {Marco}\ \bibnamefont
  {Piani}}, \bibinfo {author} {\bibfnamefont {Nathaniel}\ \bibnamefont
  {Johnston}}, \ and\ \bibinfo {author} {\bibfnamefont {Gerardo}\ \bibnamefont
  {Adesso}},\ }\bibfield  {title} {\enquote {\bibinfo {title} {Robustness of
  coherence: An operational and observable measure of quantum coherence},}\
  }\href {\doibase 10.1103/PhysRevLett.116.150502} {\bibfield  {journal}
  {\bibinfo  {journal} {Phys. Rev. Lett.}\ }\textbf {\bibinfo {volume} {116}},\
  \bibinfo {pages} {150502} (\bibinfo {year} {2016})}\BibitemShut {NoStop}%
\bibitem [{\citenamefont {Bu}\ \emph {et~al.}(2017{\natexlab{a}})\citenamefont
  {Bu}, \citenamefont {Singh}, \citenamefont {Fei}, \citenamefont {Pati},\ and\
  \citenamefont {Wu}}]{BuPRL2017}%
  \BibitemOpen
  \bibfield  {author} {\bibinfo {author} {\bibfnamefont {Kaifeng}\ \bibnamefont
  {Bu}}, \bibinfo {author} {\bibfnamefont {Uttam}\ \bibnamefont {Singh}},
  \bibinfo {author} {\bibfnamefont {Shao-Ming}\ \bibnamefont {Fei}}, \bibinfo
  {author} {\bibfnamefont {Arun~Kumar}\ \bibnamefont {Pati}}, \ and\ \bibinfo
  {author} {\bibfnamefont {Junde}\ \bibnamefont {Wu}},\ }\bibfield  {title}
  {\enquote {\bibinfo {title} {Maximum relative entropy of coherence: An
  operational coherence measure},}\ }\href {\doibase
  10.1103/PhysRevLett.119.150405} {\bibfield  {journal} {\bibinfo  {journal}
  {Phys. Rev. Lett.}\ }\textbf {\bibinfo {volume} {119}},\ \bibinfo {pages}
  {150405} (\bibinfo {year} {2017}{\natexlab{a}})}\BibitemShut {NoStop}%
\bibitem [{\citenamefont {Bae}\ \emph {et~al.}(2019)\citenamefont {Bae},
  \citenamefont {Chru\ifmmode \acute{s}\else
  \'{s}\fi{}ci\ifmmode~\acute{n}\else \'{n}\fi{}ski},\ and\ \citenamefont
  {Piani}}]{PhysRevLett.122.140404}%
  \BibitemOpen
  \bibfield  {author} {\bibinfo {author} {\bibfnamefont {Joonwoo}\ \bibnamefont
  {Bae}}, \bibinfo {author} {\bibfnamefont {Dariusz}\ \bibnamefont
  {Chru\ifmmode \acute{s}\else \'{s}\fi{}ci\ifmmode~\acute{n}\else
  \'{n}\fi{}ski}}, \ and\ \bibinfo {author} {\bibfnamefont {Marco}\
  \bibnamefont {Piani}},\ }\bibfield  {title} {\enquote {\bibinfo {title} {More
  entanglement implies higher performance in channel discrimination tasks},}\
  }\href {\doibase 10.1103/PhysRevLett.122.140404} {\bibfield  {journal}
  {\bibinfo  {journal} {Phys. Rev. Lett.}\ }\textbf {\bibinfo {volume} {122}},\
  \bibinfo {pages} {140404} (\bibinfo {year} {2019})}\BibitemShut {NoStop}%
\bibitem [{\citenamefont {Helstrom}(1976)}]{Helstrom76}%
  \BibitemOpen
  \bibfield  {author} {\bibinfo {author} {\bibfnamefont {C.W.}\ \bibnamefont
  {Helstrom}},\ }\href@noop {} {\emph {\bibinfo {title} {Quantum Detection and
  Estimation Theory}}}\ (\bibinfo  {publisher} {Academic Press, New York},\
  \bibinfo {year} {1976})\BibitemShut {NoStop}%
\bibitem [{\citenamefont {Giovannetti}\ \emph {et~al.}(2011)\citenamefont
  {Giovannetti}, \citenamefont {Lloyd},\ and\ \citenamefont
  {Maccone}}]{Giovannetti2011}%
  \BibitemOpen
  \bibfield  {author} {\bibinfo {author} {\bibfnamefont {Vittorio}\
  \bibnamefont {Giovannetti}}, \bibinfo {author} {\bibfnamefont {Seth}\
  \bibnamefont {Lloyd}}, \ and\ \bibinfo {author} {\bibfnamefont {Lorenzo}\
  \bibnamefont {Maccone}},\ }\bibfield  {title} {\enquote {\bibinfo {title}
  {Advances in quantum metrology},}\ }\href
  {http://dx.doi.org/10.1038/nphoton.2011.35} {\bibfield  {journal} {\bibinfo
  {journal} {Nat. Photon.}\ }\textbf {\bibinfo {volume} {5}},\ \bibinfo {pages}
  {222--229} (\bibinfo {year} {2011})}\BibitemShut {NoStop}%
\bibitem [{\citenamefont {Lostaglio}\ \emph
  {et~al.}(2015{\natexlab{a}})\citenamefont {Lostaglio}, \citenamefont
  {Korzekwa}, \citenamefont {Jennings},\ and\ \citenamefont
  {Rudolph}}]{Lostaglio2015}%
  \BibitemOpen
  \bibfield  {author} {\bibinfo {author} {\bibfnamefont {Matteo}\ \bibnamefont
  {Lostaglio}}, \bibinfo {author} {\bibfnamefont {Kamil}\ \bibnamefont
  {Korzekwa}}, \bibinfo {author} {\bibfnamefont {David}\ \bibnamefont
  {Jennings}}, \ and\ \bibinfo {author} {\bibfnamefont {Terry}\ \bibnamefont
  {Rudolph}},\ }\bibfield  {title} {\enquote {\bibinfo {title} {Quantum
  coherence, time-translation symmetry, and thermodynamics},}\ }\href {\doibase
  10.1103/PhysRevX.5.021001} {\bibfield  {journal} {\bibinfo  {journal} {Phys.
  Rev. X}\ }\textbf {\bibinfo {volume} {5}},\ \bibinfo {pages} {021001}
  (\bibinfo {year} {2015}{\natexlab{a}})}\BibitemShut {NoStop}%
\bibitem [{\citenamefont {Lostaglio}\ \emph
  {et~al.}(2015{\natexlab{b}})\citenamefont {Lostaglio}, \citenamefont
  {Jennings},\ and\ \citenamefont {Rudolph}}]{Lostaglio2015NC}%
  \BibitemOpen
  \bibfield  {author} {\bibinfo {author} {\bibfnamefont {Matteo}\ \bibnamefont
  {Lostaglio}}, \bibinfo {author} {\bibfnamefont {David}\ \bibnamefont
  {Jennings}}, \ and\ \bibinfo {author} {\bibfnamefont {Terry}\ \bibnamefont
  {Rudolph}},\ }\bibfield  {title} {\enquote {\bibinfo {title} {Description of
  quantum coherence in thermodynamic processes requires constraints beyond free
  energy},}\ }\href {http://dx.doi.org/10.1038/ncomms7383} {\bibfield
  {journal} {\bibinfo  {journal} {Nat. Commun.}\ }\textbf {\bibinfo {volume}
  {6}} (\bibinfo {year} {2015}{\natexlab{b}})}\BibitemShut {NoStop}%
\bibitem [{\citenamefont {Plenio}\ and\ \citenamefont
  {Huelga}(2008)}]{Plenio08}%
  \BibitemOpen
  \bibfield  {author} {\bibinfo {author} {\bibfnamefont {M.~B.}\ \bibnamefont
  {Plenio}}\ and\ \bibinfo {author} {\bibfnamefont {S.~F.}\ \bibnamefont
  {Huelga}},\ }\bibfield  {title} {\enquote {\bibinfo {title}
  {Dephasing-assisted transport: quantum networks and biomolecules},}\ }\href
  {http://stacks.iop.org/1367-2630/10/i=11/a=113019} {\bibfield  {journal}
  {\bibinfo  {journal} {New J. Phys.}\ }\textbf {\bibinfo {volume} {10}},\
  \bibinfo {pages} {113019} (\bibinfo {year} {2008})}\BibitemShut {NoStop}%
\bibitem [{\citenamefont {Levi}\ and\ \citenamefont {Mintert}(2014)}]{Levi14}%
  \BibitemOpen
  \bibfield  {author} {\bibinfo {author} {\bibfnamefont {Federico}\
  \bibnamefont {Levi}}\ and\ \bibinfo {author} {\bibfnamefont {Florian}\
  \bibnamefont {Mintert}},\ }\bibfield  {title} {\enquote {\bibinfo {title} {A
  quantitative theory of coherent delocalization},}\ }\href
  {http://stacks.iop.org/1367-2630/16/i=3/a=033007} {\bibfield  {journal}
  {\bibinfo  {journal} {New J. Phys.}\ }\textbf {\bibinfo {volume} {16}},\
  \bibinfo {pages} {033007} (\bibinfo {year} {2014})}\BibitemShut {NoStop}%
\bibitem [{\citenamefont {Hu}\ \emph {et~al.}(2018)\citenamefont {Hu},
  \citenamefont {Hu}, \citenamefont {Wang}, \citenamefont {Peng}, \citenamefont
  {Zhang},\ and\ \citenamefont {Fan}}]{Hu2018}%
  \BibitemOpen
  \bibfield  {author} {\bibinfo {author} {\bibfnamefont {Ming-Liang}\
  \bibnamefont {Hu}}, \bibinfo {author} {\bibfnamefont {Xueyuan}\ \bibnamefont
  {Hu}}, \bibinfo {author} {\bibfnamefont {Jieci}\ \bibnamefont {Wang}},
  \bibinfo {author} {\bibfnamefont {Yi}~\bibnamefont {Peng}}, \bibinfo {author}
  {\bibfnamefont {Yu-Ran}\ \bibnamefont {Zhang}}, \ and\ \bibinfo {author}
  {\bibfnamefont {Heng}\ \bibnamefont {Fan}},\ }\bibfield  {title} {\enquote
  {\bibinfo {title} {Quantum coherence and geometric quantum discord},}\ }\href
  {\doibase https://doi.org/10.1016/j.physrep.2018.07.004} {\bibfield
  {journal} {\bibinfo  {journal} {Phys. Rep.}\ }\textbf {\bibinfo {volume}
  {762-764}},\ \bibinfo {pages} {1 -- 100} (\bibinfo {year}
  {2018})}\BibitemShut {NoStop}%
\bibitem [{\citenamefont {Chitambar}\ and\ \citenamefont
  {Gour}(2016{\natexlab{a}})}]{Chitambar2016b}%
  \BibitemOpen
  \bibfield  {author} {\bibinfo {author} {\bibfnamefont {Eric}\ \bibnamefont
  {Chitambar}}\ and\ \bibinfo {author} {\bibfnamefont {Gilad}\ \bibnamefont
  {Gour}},\ }\bibfield  {title} {\enquote {\bibinfo {title} {Comparison of
  incoherent operations and measures of coherence},}\ }\href {\doibase
  10.1103/PhysRevA.94.052336} {\bibfield  {journal} {\bibinfo  {journal} {Phys.
  Rev. A}\ }\textbf {\bibinfo {volume} {94}},\ \bibinfo {pages} {052336}
  (\bibinfo {year} {2016}{\natexlab{a}})}\BibitemShut {NoStop}%
\bibitem [{\citenamefont {Chitambar}\ and\ \citenamefont
  {Gour}(2016{\natexlab{b}})}]{Chitambar2016a}%
  \BibitemOpen
  \bibfield  {author} {\bibinfo {author} {\bibfnamefont {Eric}\ \bibnamefont
  {Chitambar}}\ and\ \bibinfo {author} {\bibfnamefont {Gilad}\ \bibnamefont
  {Gour}},\ }\bibfield  {title} {\enquote {\bibinfo {title} {Critical
  examination of incoherent operations and a physically consistent resource
  theory of quantum coherence},}\ }\href {\doibase
  10.1103/PhysRevLett.117.030401} {\bibfield  {journal} {\bibinfo  {journal}
  {Phys. Rev. Lett.}\ }\textbf {\bibinfo {volume} {117}},\ \bibinfo {pages}
  {030401} (\bibinfo {year} {2016}{\natexlab{b}})}\BibitemShut {NoStop}%
\bibitem [{\citenamefont {Shao}\ \emph {et~al.}(2015)\citenamefont {Shao},
  \citenamefont {Xi}, \citenamefont {Fan},\ and\ \citenamefont
  {Li}}]{Shao2015}%
  \BibitemOpen
  \bibfield  {author} {\bibinfo {author} {\bibfnamefont {Lian-He}\ \bibnamefont
  {Shao}}, \bibinfo {author} {\bibfnamefont {Zhengjun}\ \bibnamefont {Xi}},
  \bibinfo {author} {\bibfnamefont {Heng}\ \bibnamefont {Fan}}, \ and\ \bibinfo
  {author} {\bibfnamefont {Yongming}\ \bibnamefont {Li}},\ }\bibfield  {title}
  {\enquote {\bibinfo {title} {Fidelity and trace-norm distances for
  quantifying coherence},}\ }\href {\doibase 10.1103/PhysRevA.91.042120}
  {\bibfield  {journal} {\bibinfo  {journal} {Phys. Rev. A}\ }\textbf {\bibinfo
  {volume} {91}},\ \bibinfo {pages} {042120} (\bibinfo {year}
  {2015})}\BibitemShut {NoStop}%
\bibitem [{\citenamefont {Piani}\ \emph {et~al.}(2016)\citenamefont {Piani},
  \citenamefont {Cianciaruso}, \citenamefont {Bromley}, \citenamefont {Napoli},
  \citenamefont {Johnston},\ and\ \citenamefont {Adesso}}]{Piani2016}%
  \BibitemOpen
  \bibfield  {author} {\bibinfo {author} {\bibfnamefont {Marco}\ \bibnamefont
  {Piani}}, \bibinfo {author} {\bibfnamefont {Marco}\ \bibnamefont
  {Cianciaruso}}, \bibinfo {author} {\bibfnamefont {Thomas~R.}\ \bibnamefont
  {Bromley}}, \bibinfo {author} {\bibfnamefont {Carmine}\ \bibnamefont
  {Napoli}}, \bibinfo {author} {\bibfnamefont {Nathaniel}\ \bibnamefont
  {Johnston}}, \ and\ \bibinfo {author} {\bibfnamefont {Gerardo}\ \bibnamefont
  {Adesso}},\ }\bibfield  {title} {\enquote {\bibinfo {title} {Robustness of
  asymmetry and coherence of quantum states},}\ }\href {\doibase
  10.1103/PhysRevA.93.042107} {\bibfield  {journal} {\bibinfo  {journal} {Phys.
  Rev. A}\ }\textbf {\bibinfo {volume} {93}},\ \bibinfo {pages} {042107}
  (\bibinfo {year} {2016})}\BibitemShut {NoStop}%
\bibitem [{\citenamefont {Rana}\ \emph {et~al.}(2016)\citenamefont {Rana},
  \citenamefont {Parashar},\ and\ \citenamefont {Lewenstein}}]{Rana2016pra}%
  \BibitemOpen
  \bibfield  {author} {\bibinfo {author} {\bibfnamefont {Swapan}\ \bibnamefont
  {Rana}}, \bibinfo {author} {\bibfnamefont {Preeti}\ \bibnamefont {Parashar}},
  \ and\ \bibinfo {author} {\bibfnamefont {Maciej}\ \bibnamefont
  {Lewenstein}},\ }\bibfield  {title} {\enquote {\bibinfo {title}
  {Trace-distance measure of coherence},}\ }\href {\doibase
  10.1103/PhysRevA.93.012110} {\bibfield  {journal} {\bibinfo  {journal} {Phys.
  Rev. A}\ }\textbf {\bibinfo {volume} {93}},\ \bibinfo {pages} {012110}
  (\bibinfo {year} {2016})}\BibitemShut {NoStop}%
\bibitem [{\citenamefont {Bu}\ \emph {et~al.}(2017{\natexlab{b}})\citenamefont
  {Bu}, \citenamefont {Kumar}, \citenamefont {Zhang},\ and\ \citenamefont
  {Wu}}]{Bu2017PLA}%
  \BibitemOpen
  \bibfield  {author} {\bibinfo {author} {\bibfnamefont {Kaifeng}\ \bibnamefont
  {Bu}}, \bibinfo {author} {\bibfnamefont {Asutosh}\ \bibnamefont {Kumar}},
  \bibinfo {author} {\bibfnamefont {Lin}\ \bibnamefont {Zhang}}, \ and\
  \bibinfo {author} {\bibfnamefont {Junde}\ \bibnamefont {Wu}},\ }\bibfield
  {title} {\enquote {\bibinfo {title} {Cohering power of quantum operations},}\
  }\href {\doibase https://doi.org/10.1016/j.physleta.2017.03.022} {\bibfield
  {journal} {\bibinfo  {journal} {Phys. Lett. A}\ }\textbf {\bibinfo {volume}
  {381}},\ \bibinfo {pages} {1670 -- 1676} (\bibinfo {year}
  {2017}{\natexlab{b}})}\BibitemShut {NoStop}%
\bibitem [{\citenamefont {D\'iaz}\ \emph {et~al.}()\citenamefont {D\'iaz},
  \citenamefont {Fang}, \citenamefont {Wang}, \citenamefont {Rosati},
  \citenamefont {Skotiniotis}, \citenamefont {Calsamiglia},\ and\ \citenamefont
  {Winter}}]{Diaz2018}%
  \BibitemOpen
  \bibfield  {author} {\bibinfo {author} {\bibfnamefont {M.G}\ \bibnamefont
  {D\'iaz}}, \bibinfo {author} {\bibfnamefont {K.}~\bibnamefont {Fang}},
  \bibinfo {author} {\bibfnamefont {X.}~\bibnamefont {Wang}}, \bibinfo {author}
  {\bibfnamefont {M.}~\bibnamefont {Rosati}}, \bibinfo {author} {\bibfnamefont
  {M}~\bibnamefont {Skotiniotis}}, \bibinfo {author} {\bibfnamefont
  {J.}~\bibnamefont {Calsamiglia}}, \ and\ \bibinfo {author} {\bibfnamefont
  {A.}~\bibnamefont {Winter}},\ }\href@noop {} {}\Eprint
  {http://arxiv.org/abs/1805.04045} {arXiv:1805.04045} \BibitemShut {NoStop}%
\bibitem [{\citenamefont {Yuan}()}]{Yuan2018}%
  \BibitemOpen
  \bibfield  {author} {\bibinfo {author} {\bibfnamefont {Xiao}\ \bibnamefont
  {Yuan}},\ }\bibfield  {title} {\enquote {\bibinfo {title} {Relative entropies
  of quantum channels with applications in resource theory},}\ }\href@noop {}
  {\ }\Eprint {http://arxiv.org/abs/1807.05958} {arXiv:1807.05958} \BibitemShut
  {NoStop}%
\bibitem [{\citenamefont {Chen}\ \emph {et~al.}(2016)\citenamefont {Chen},
  \citenamefont {Grogan}, \citenamefont {Johnston}, \citenamefont {Li},\ and\
  \citenamefont {Plosker}}]{Chen2016}%
  \BibitemOpen
  \bibfield  {author} {\bibinfo {author} {\bibfnamefont {Jianxin}\ \bibnamefont
  {Chen}}, \bibinfo {author} {\bibfnamefont {Shane}\ \bibnamefont {Grogan}},
  \bibinfo {author} {\bibfnamefont {Nathaniel}\ \bibnamefont {Johnston}},
  \bibinfo {author} {\bibfnamefont {Chi-Kwong}\ \bibnamefont {Li}}, \ and\
  \bibinfo {author} {\bibfnamefont {Sarah}\ \bibnamefont {Plosker}},\
  }\bibfield  {title} {\enquote {\bibinfo {title} {Quantifying the coherence of
  pure quantum states},}\ }\href {\doibase 10.1103/PhysRevA.94.042313}
  {\bibfield  {journal} {\bibinfo  {journal} {Phys. Rev. A}\ }\textbf {\bibinfo
  {volume} {94}},\ \bibinfo {pages} {042313} (\bibinfo {year}
  {2016})}\BibitemShut {NoStop}%
\bibitem [{\citenamefont {Yu}\ \emph {et~al.}(2016)\citenamefont {Yu},
  \citenamefont {Zhang}, \citenamefont {Xu},\ and\ \citenamefont
  {Tong}}]{Yu2016}%
  \BibitemOpen
  \bibfield  {author} {\bibinfo {author} {\bibfnamefont {Xiao-Dong}\
  \bibnamefont {Yu}}, \bibinfo {author} {\bibfnamefont {Da-Jian}\ \bibnamefont
  {Zhang}}, \bibinfo {author} {\bibfnamefont {G.~F.}\ \bibnamefont {Xu}}, \
  and\ \bibinfo {author} {\bibfnamefont {D.~M.}\ \bibnamefont {Tong}},\
  }\bibfield  {title} {\enquote {\bibinfo {title} {Alternative framework for
  quantifying coherence},}\ }\href {\doibase 10.1103/PhysRevA.94.060302}
  {\bibfield  {journal} {\bibinfo  {journal} {Phys. Rev. A}\ }\textbf {\bibinfo
  {volume} {94}},\ \bibinfo {pages} {060302} (\bibinfo {year}
  {2016})}\BibitemShut {NoStop}%
\bibitem [{\citenamefont {Wang}\ \emph {et~al.}(2017)\citenamefont {Wang},
  \citenamefont {Tang}, \citenamefont {Wei}, \citenamefont {Yu}, \citenamefont
  {Ke}, \citenamefont {Xu}, \citenamefont {Li},\ and\ \citenamefont
  {Guo}}]{Wang2017}%
  \BibitemOpen
  \bibfield  {author} {\bibinfo {author} {\bibfnamefont {Yi-Tao}\ \bibnamefont
  {Wang}}, \bibinfo {author} {\bibfnamefont {Jian-Shun}\ \bibnamefont {Tang}},
  \bibinfo {author} {\bibfnamefont {Zhi-Yuan}\ \bibnamefont {Wei}}, \bibinfo
  {author} {\bibfnamefont {Shang}\ \bibnamefont {Yu}}, \bibinfo {author}
  {\bibfnamefont {Zhi-Jin}\ \bibnamefont {Ke}}, \bibinfo {author}
  {\bibfnamefont {Xiao-Ye}\ \bibnamefont {Xu}}, \bibinfo {author}
  {\bibfnamefont {Chuan-Feng}\ \bibnamefont {Li}}, \ and\ \bibinfo {author}
  {\bibfnamefont {Guang-Can}\ \bibnamefont {Guo}},\ }\bibfield  {title}
  {\enquote {\bibinfo {title} {Directly measuring the degree of quantum
  coherence using interference fringes},}\ }\href {\doibase
  10.1103/PhysRevLett.118.020403} {\bibfield  {journal} {\bibinfo  {journal}
  {Phys. Rev. Lett.}\ }\textbf {\bibinfo {volume} {118}},\ \bibinfo {pages}
  {020403} (\bibinfo {year} {2017})}\BibitemShut {NoStop}%
\bibitem [{\citenamefont {Zheng}\ \emph {et~al.}(2018)\citenamefont {Zheng},
  \citenamefont {Ma}, \citenamefont {Wang}, \citenamefont {Fei},\ and\
  \citenamefont {Peng}}]{ZhengPRL2018}%
  \BibitemOpen
  \bibfield  {author} {\bibinfo {author} {\bibfnamefont {Wenqiang}\
  \bibnamefont {Zheng}}, \bibinfo {author} {\bibfnamefont {Zhihao}\
  \bibnamefont {Ma}}, \bibinfo {author} {\bibfnamefont {Hengyan}\ \bibnamefont
  {Wang}}, \bibinfo {author} {\bibfnamefont {Shao-Ming}\ \bibnamefont {Fei}}, \
  and\ \bibinfo {author} {\bibfnamefont {Xinhua}\ \bibnamefont {Peng}},\
  }\bibfield  {title} {\enquote {\bibinfo {title} {Experimental demonstration
  of observability and operability of robustness of coherence},}\ }\href
  {\doibase 10.1103/PhysRevLett.120.230504} {\bibfield  {journal} {\bibinfo
  {journal} {Phys. Rev. Lett.}\ }\textbf {\bibinfo {volume} {120}},\ \bibinfo
  {pages} {230504} (\bibinfo {year} {2018})}\BibitemShut {NoStop}%
\bibitem [{\citenamefont {{Liu}}\ and\ \citenamefont {{Yuan}}(2019)}]{LiuYuan}%
  \BibitemOpen
  \bibfield  {author} {\bibinfo {author} {\bibfnamefont {Yunchao}\ \bibnamefont
  {{Liu}}}\ and\ \bibinfo {author} {\bibfnamefont {Xiao}\ \bibnamefont
  {{Yuan}}},\ }\bibfield  {title} {\enquote {\bibinfo {title} {{Operational
  Resource Theory of Quantum Channels}},}\ }\href@noop {} {\bibfield  {journal}
  {\bibinfo  {journal} {arXiv e-prints}\ ,\ \bibinfo {eid} {arXiv:1904.02680}}
  (\bibinfo {year} {2019})},\ \Eprint {http://arxiv.org/abs/1904.02680}
  {arXiv:1904.02680 [quant-ph]} \BibitemShut {NoStop}%
\bibitem [{\citenamefont {Datta}(2009)}]{Datta2009IEEE}%
  \BibitemOpen
  \bibfield  {author} {\bibinfo {author} {\bibfnamefont {N.}~\bibnamefont
  {Datta}},\ }\bibfield  {title} {\enquote {\bibinfo {title} {Min- and
  max-relative entropies and a new entanglement monotone},}\ }\href {\doibase
  10.1109/TIT.2009.2018325} {\bibfield  {journal} {\bibinfo  {journal} {IEEE
  Trans. Inf. Theory}\ }\textbf {\bibinfo {volume} {55}},\ \bibinfo {pages}
  {2816--2826} (\bibinfo {year} {2009})}\BibitemShut {NoStop}%
\bibitem [{\citenamefont {Barnum}\ \emph {et~al.}(1996)\citenamefont {Barnum},
  \citenamefont {Caves}, \citenamefont {Fuchs}, \citenamefont {Jozsa},\ and\
  \citenamefont {Schumacher}}]{Barnum96}%
  \BibitemOpen
  \bibfield  {author} {\bibinfo {author} {\bibfnamefont {Howard}\ \bibnamefont
  {Barnum}}, \bibinfo {author} {\bibfnamefont {Carlton~M.}\ \bibnamefont
  {Caves}}, \bibinfo {author} {\bibfnamefont {Christopher~A.}\ \bibnamefont
  {Fuchs}}, \bibinfo {author} {\bibfnamefont {Richard}\ \bibnamefont {Jozsa}},
  \ and\ \bibinfo {author} {\bibfnamefont {Benjamin}\ \bibnamefont
  {Schumacher}},\ }\bibfield  {title} {\enquote {\bibinfo {title} {Noncommuting
  mixed states cannot be broadcast},}\ }\href {\doibase
  10.1103/PhysRevLett.76.2818} {\bibfield  {journal} {\bibinfo  {journal}
  {Phys. Rev. Lett.}\ }\textbf {\bibinfo {volume} {76}},\ \bibinfo {pages}
  {2818--2821} (\bibinfo {year} {1996})}\BibitemShut {NoStop}%
\bibitem [{\citenamefont {Gilchrist}\ \emph {et~al.}(2005)\citenamefont
  {Gilchrist}, \citenamefont {Langford},\ and\ \citenamefont
  {Nielsen}}]{Gilchrist05}%
  \BibitemOpen
  \bibfield  {author} {\bibinfo {author} {\bibfnamefont {Alexei}\ \bibnamefont
  {Gilchrist}}, \bibinfo {author} {\bibfnamefont {Nathan~K.}\ \bibnamefont
  {Langford}}, \ and\ \bibinfo {author} {\bibfnamefont {Michael~A.}\
  \bibnamefont {Nielsen}},\ }\bibfield  {title} {\enquote {\bibinfo {title}
  {Distance measures to compare real and ideal quantum processes},}\ }\href
  {\doibase 10.1103/PhysRevA.71.062310} {\bibfield  {journal} {\bibinfo
  {journal} {Phys. Rev. A}\ }\textbf {\bibinfo {volume} {71}},\ \bibinfo
  {pages} {062310} (\bibinfo {year} {2005})}\BibitemShut {NoStop}%
\bibitem [{\citenamefont {Uhlmann}(1976)}]{Uhlmann76}%
  \BibitemOpen
  \bibfield  {author} {\bibinfo {author} {\bibfnamefont {A.}~\bibnamefont
  {Uhlmann}},\ }\bibfield  {title} {\enquote {\bibinfo {title} {Noncommuting
  mixed states cannot be broadcast},}\ }\href@noop {} {\bibfield  {journal}
  {\bibinfo  {journal} {Rep. Math. Phys.}\ }\textbf {\bibinfo {volume} {9}},\
  \bibinfo {pages} {273} (\bibinfo {year} {1976})}\BibitemShut {NoStop}%
\bibitem [{\citenamefont {Bu}\ and\ \citenamefont {Xiong}(2017)}]{Bu2016QIC}%
  \BibitemOpen
  \bibfield  {author} {\bibinfo {author} {\bibfnamefont {Kaifeng}\ \bibnamefont
  {Bu}}\ and\ \bibinfo {author} {\bibfnamefont {Chunhe}\ \bibnamefont
  {Xiong}},\ }\bibfield  {title} {\enquote {\bibinfo {title} {A note on
  cohering power and de-cohering power},}\ }\href {\doibase
  10.26421/QIC17.13-14} {\bibfield  {journal} {\bibinfo  {journal} {Quan. Inf.
  Comp.}\ }\textbf {\bibinfo {volume} {13}},\ \bibinfo {pages} {1206--1220}
  (\bibinfo {year} {2017})}\BibitemShut {NoStop}%
\end{thebibliography}%

\appendix

\section{Connections between channel discrimination  and resource 
generating/increasing  power }
\label{apen:1}

Given a distance measure $D: \cD(\cH)\times \cD(\cH)\to \real_+$,  we consider the following  conditions:

(1) Positivity:  $D(\rho, \sigma)\geq 0$, $D(\rho, \sigma)=0$ iff $\rho=\sigma$.

(2) Pseudo joint convexity: 
$D(\sum_ip_i\rho_i, \sum_ip_i\sigma_i)\leq \max_i D(\rho_i, \sigma_i)$ with $\sum_ip_i=1$.

(2') Joint convexity:
$D(\sum_ip_i\rho_i, \sum_ip_i\sigma_i)\leq \sum_i p_iD(\rho_i, \sigma_i)$ with $\sum_ip_i=1$.

(3) Data processing inequality:
$D(\cN(\rho), \cN(\sigma))\leq D(\rho, \sigma)$ for any CPTP map $\cN$.

(4) Triangle inequality: 
$D(\rho, \sigma)\leq D(\rho, \tau)+D(\tau, \sigma)$
for any $\tau\in \cD(\cH)$.

Here, we assume the distance measure always satisfies the condition (1) ,
i.e., positivity.

\begin{lem}\label{lem:5}
For any given distance measure $D$ and quantum channel 
$\cN$, it holds that 
\begin{eqnarray}
 \Omega_D(\cN) =\max_{\rho\in \mathcal{F}}\min_{\cM\in \mathfrak{F}}D(\cN(\rho), \cM(\rho)).
\end{eqnarray}

\end{lem}

\begin{proof}
First,  we have
\begin{align*}
  &\max_{\rho\in \mathcal{F}}\min_{\cM\in \mathfrak{F}}D(\cN(\rho), \cM(\rho))\\
  \geq&\max_{\rho\in \mathcal{F}}\min_{\sigma\in \mathcal{F}}D(\cN(\rho), \sigma)\\
  =&\max_{\rho\in \mathcal{F}}\omega_{D}(\cN(\rho))\\
  =&\Omega_{D}(\cN),
\end{align*}
where the inequality comes from the fact that $\cM(\rho)\in\mathcal{F}$ for any $\rho\in \mathcal{F}$.

Besides, for any $\rho\in \mathcal{F}$, we can define 
the quantum channel $\cN_\rho$  as $\cN_{\rho}(\tau)=\sigma^{\star}_{\cN(\rho)}$ for any quantum
state $\tau\in \cD(\cH)$ with $\sigma^{\star}_{\cN(\rho)}\in \mathcal{F}$ and  
$\omega_{D}(\cN(\rho))=D(\cN(\rho), \sigma^{\star}_{\cN(\rho)})$. 
It is easy to verify that $\cN_\rho$ is a free operation,  i.e.,  $\cN_\rho\in \mathfrak{F}$.  Thus, 

\begin{eqnarray*}
  &&\max_{\rho\in \mathcal{F}}\min_{\cM\in \mathfrak{F}}D(\cN(\rho), \cM(\rho))\\
  &\leq&\max_{\rho\in \mathcal{F}} D(\cN(\rho), \cN_{\rho}(\rho))\\
    &=&\max_{\rho\in \mathcal{F}}D(\cN(\rho), \sigma^{\star}_{\cN(\rho)})\\
  &=&\max_{\rho\in \mathcal{F}}\omega_{D}(\cN(\rho))\\
  &=&\Omega_{D}(\cN),
\end{eqnarray*}
where the inequality comes from the fact that $\cN_{\rho}\in\mathfrak{F}$ and 
$\cN_{\rho} $ maps any quantum state to the free state $\sigma^{\star}_{\cN(\rho)}$.

\end{proof}

\begin{lem}\label{lem:inc}
If the distance measure $D$ satisfies the 
triangle inequality and the data processing inequality ( i.e., non-increasing under CPTP maps), then we have
\begin{eqnarray}
\Omega_D(\cN)
=\widetilde{\Omega}_D(\cN).
\end{eqnarray}
\end{lem}
\begin{proof}
It is obvious that $\Omega_{D}\leq\widetilde{\Omega}_{D}$,
thus we only need to prove $\widetilde{\Omega}_{D}\leq\Omega_{D}$.

For any quantum state $\rho\in\cD(\cH)$, we have
\begin{eqnarray*}
  && \omega_{D}(\cN(\rho))-\omega_{D}(\rho)\\
  &=&\min_{\sigma\in \mathcal{F}}D(\cN(\rho), \sigma)
  -\min_{\tau\in \mathcal{F}}D(\rho, \tau)\\
  &=&\max_{\tau\in \mathcal{F}}[\min_{\sigma\in \mathcal{F}}(D(\cN(\rho), \sigma)-D(\rho, \tau))]\\
  &\leq&\max_{\tau\in \mathcal{F}}\min_{\sigma\in \mathcal{F}}[D(\cN(\rho), \sigma)-D(\cN(\rho), \cN(\tau))]\\
  &\leq&\max_{\tau\in \mathcal{F}}\min_{\sigma\in \mathcal{F}}D(\cN(\tau), \sigma)\\
  &=&\max_{\tau\in \mathcal{F}}\omega_{D}(\cN(\tau))\\
  &=&\Omega_{D}(\cN),
\end{eqnarray*}
where the first inequality comes from  the 
data processing inequality
and the  second inequality comes from the triangle inequality of $D$. 
Therefore, we have $\widetilde{\Omega}_{D}(\cN)\leq\Omega_{D}(\cN)$.

\end{proof}

\begin{mproof}[ Proof of Theorem 2]
It is easy to verify that 
trace-norm satisfies the data processing inequality  and the triangle 
inequality. 
Thus, according to Lemma \ref{lem:5} and \ref{lem:inc}, we have 
\begin{eqnarray*}
\widetilde{\Omega}_1(\cN)=
 \Omega_1(\cN) =\frac{1}{2}\max_{\rho\in \mathcal{F}}\min_{\cM\in \mathfrak{F}}\norm{\cN(\rho)- \cM(\rho)}_1.
\end{eqnarray*}
Besides, the success probability $p_{\rm succ}(\cN,\mathfrak{F}, \mathcal{F})$
can be expressed as
\begin{eqnarray*}
p_{\rm succ}(\cN,\mathfrak{F}, \mathcal{F})&=&\frac{1}{2}+\frac{1}{4}\max_{\rho\in \mathcal{F}}\min_{\cM\in \mathfrak{F}}\norm{\cN(\rho)- \cM(\rho)}_1\\
&=&\frac{1}{2}+\frac{1}{2}\widetilde{\Omega}_1(\cN)\\
&=&\frac{1}{2}+\frac{1}{2}\Omega_1(\cN).
\end{eqnarray*}
\end{mproof}

\begin{cor}
If we take the distance measure $D$ to be max-relative entropy $D_{\max}$ or fidelity $D_F$, then we have 

\begin{eqnarray}
\widetilde{\Omega}_D(\cN)=
 \Omega_D(\cN) =\max_{\rho\in \mathcal{F}}\min_{\cM\in \mathfrak{F}}D(\cN(\rho), \cM(\rho)),
\end{eqnarray}
where $D_F(\rho, \sigma)
=\sqrt{1-F^2(\rho, \sigma)}$ with 
$F(\rho, \sigma)=\Tr{|\sqrt{\rho}\sqrt{\sigma}|}$.
\end{cor}
\begin{proof}
It has been proved that $D_{\max}$ satisfies the 
data processing inequality \cite{Datta2009IEEE} and
the triangle inequality comes directly from the definitions.
Besides, it has been proved that 
$D_F$ satisfies the  data processing inequality \cite{Barnum96} and
the triangle inequality \cite{Gilchrist05,Uhlmann76}.
\end{proof}

Now, let us consider the example of  coherence.
In single-qubit system, it has been proved that 
trace-norm of coherence
$C_1$ is equivalent to $l_1$ norm of coherence $C_{l_1}$ \cite{Rana2016pra,Shao2015} and 
the analytic form of coherence generating power for unitary operations has been 
obtained in \cite{Bu2017PLA}. Therefore, we have the following corollary,
\begin{cor}
Given a single-qubit unitary $U=[U_{ij}]_{i,j=1,2}$, the coherence generating power by  trace-norm
is 
\begin{eqnarray}
\mathcal{C}_1(U)=\max_{i=1,2} |U_{i1}U_{i2}|. 
\end{eqnarray}
Specially, for the Hadamard gate $H$, $\mathcal{C}_1(H)=1/2$.
\end{cor}

\section{Properties of $\Omega_{D}(\cN)$ }\label{apen:2}

Now, let us investigate the properties of 
$\Omega_D(\cN)$ for any distance measure  $D$. We assume that 
the free states on $\mathcal{H}_A\otimes \mathcal{H}_B$ 
is defined as convex combination of the tensor product of 
free states on $\mathcal{H}_A$ and $\mathcal{H}_B$, i.e.,
$\mathcal{F}_{AB}= Conv\set{\mathcal{F}_A\otimes \mathcal{F}_B }$.

\begin{lem}\label{prop:proD}
 Given any distance measure $D$, $\Omega_{D}(\cdot)$
 has the following properties: 
 
(i) $\Omega_{D}(\cN)\geq 0$, and $\Omega_{D}(\cN)=0$ if $\cN\in \mathfrak{F}$. Moreover, if $\mathfrak{F}$ includes 
all CPTP maps which maps all free states to free states, then $\Omega_{D}(\cN)=0$ iff $\cN\in \mathfrak{F}$.
 
(ii)  If the distance measure $D$ satisfies the data processing inequality: For any  $\cM_1, \cM_2 \in \mathfrak{F}$,
\begin{eqnarray}
 \Omega_{D}(\cM_1\circ \cN\circ \cM_2)
\leq  \Omega_{D}(\cN).
\end{eqnarray}

 (iii) If the distance measure $D$ satisfies joint convexity: 
 Given a set of quantum channels $\set{\cN_i, p_i}_i$ with $\sum_ip_i=1$, 
 \begin{eqnarray}
  \Omega_{D}(\sum_i p_i\cN_i)
\leq  \sum_i p_i\Omega_{D}(\cN_i).
 \end{eqnarray}

(iv) If the distance measure $D$ satisfies the pseudo joint convexity and  data processing inequality: 
 Given two channels $\cN_1$ and $\cN_2$, it holds that 
\begin{eqnarray}
\Omega_{D}(\cN_1\ot\cN_2)\geq \max\set{\Omega_{D}(\cN_1), \Omega_{D}(\cN_2)}.
\end{eqnarray}

(v)  If the distance measure $D$ satisfies the pseudo joint convexity, data processing inequality
and  triangle inequality: 
Given two channels $\cN_1$ and $\cN_2$, it holds that 
\begin{eqnarray}
 \Omega_{D}(\cN_1\ot\cN_2) \leq \Omega_{D}(\cN_1)+ \Omega_{D}(\cN_2).
 \end{eqnarray}

 \end{lem}

 \begin{proof}
 (i) This comes directly from the definition.
   
  (ii)
 For any $\cM\in \mathfrak{F}$,
 \begin{eqnarray*}
&&  \Omega_{D}(\cM\circ \cN)\\
&=&\max_{\rho\in \mathcal{F}}\min_{\cM\in \mathfrak{F}} D(\cM\circ\cN(\rho), \cM(\rho))\\
&\leq& \max_{\rho\in \mathcal{F}}\min_{\cM\in \mathfrak{F}}  D(\cM\circ\cN(\rho), \cM\circ\cM(\rho))\\
&\leq&  \max_{\rho\in \mathcal{F}}\min_{\cM\in \mathfrak{F}}D(\cN(\rho), \cM(\rho))\\
&=& \Omega_{D}( \cN), 
 \end{eqnarray*}
 where the first inequality comes from the fact that 
$\cM\circ \cM \in \mathfrak{F}$ for any $\cM\in \mathfrak{F}$ and the second inequality 
comes from the data processing inequality.

Besides,  
\begin{eqnarray*}
 \Omega_{D}( \cN\circ\cM)
 =\max_{\rho\in \mathcal{F}}
 \omega_D(\cN(\cM(\rho)))
 \leq \max_{\rho\in \mathcal{F}}
  \omega_D(\cN(\rho)),
\end{eqnarray*}
where the inequality comes from the fact 
$\cM(\mathcal{F})\subset\mathcal{F}$.
 \end{proof}

 (iii)
 Since $D$ is jointly convex, then the corresponding resource monotone
 $\omega_D$ is convex, i.e., $\omega_D(\sum_ip_i\rho_i)\leq \sum_ip_i\omega_D(\rho_i)$.
  Thus, 
   
 \begin{eqnarray*}
&& \Omega_D(\sum_ip_i\cN_i)\\
&=&\max_{\rho\in\mathcal{F}}\omega_D(\sum_i p_i\cN_i(\rho))\\
& \leq&\max_{\rho\in\mathcal{F}}\sum_ip_i\omega_D(\cN_i(\rho))\\
 &\leq& \sum_ip_i\max_{\rho\in\mathcal{F}}\omega_D(\cN_i(\rho))\\
 &=&\sum_ip_i\Omega_D(\cN_i).
 \end{eqnarray*}

 (iv) 
We only need 
to prove that 
$\max\set{\omega_D(\rho_1), \omega_D(\rho_2)}\leq \omega_D(\rho_1\ot \rho_2)  $.

First, 
\begin{eqnarray}
\min_{\tau_{12}\in \mathcal{F}_{12}}
D(\rho_1\ot\rho_2, \tau_{12})
\geq \min_{\tau_1\in \mathcal{F}_1}D(\rho_1, \tau_1),
\end{eqnarray}
 where $\tau_1=\Ptr{2}{\tau_{12}}$ and 
 the inequality comes from the data processing inequality. Hence, we have 
 $\omega_D(\rho_1\ot \rho_2)\geq \omega_D(\rho_1)$. Similarly, we have 
 $\omega_D(\rho_1\ot \rho_2)\geq \omega_D(\rho_2)$. 
 
 (v) 
We only need to prove that 
 $\omega_D(\rho_1\ot \rho_2)  \leq \omega_D(\rho_1)+\omega_D(\rho_2)$.
Due to the data processing inequality,  
we have 
\begin{eqnarray}
D(\rho, \sigma)
=D(\rho\ot \tau, \sigma\ot \tau),
\end{eqnarray}
because both partial trace and tensoring with a quantum state
are CPTP maps.

Therefore, we have
 \begin{eqnarray*}
&&\min_{\tau_{12}\in \cI}
D(\rho_1\ot\rho_2, \tau_{12})\\
&\leq& D(\rho_1\ot\rho_2, \tau_1\ot\tau_2)\\
&\leq&D(\rho_1\ot\rho_2, \tau_1\ot\rho_2)
+D(\tau_1\ot\rho_2, \tau_1\ot\tau_2)\\
&=& D(\rho_1, \tau_1)
+D(\rho_2, \tau_2)\\
&=&\omega_D(\rho_1)+\omega_D(\rho_2).
 \end{eqnarray*}
 where the free states  $\tau_1$ and $\tau_2$ are chosen to satisfy the conditions  $\omega_D(\rho_1)=D(\rho_1, \tau_1)$ and 
 $\omega_D(\rho_2)=D(\rho_2, \tau_2)$.

 \begin{mproof}[Proof of Proposition 3]
Since the trace norm satisfies the  joint convexity, data processing inequality
and  triangle inequality, then the Proposition 3
comes directly from the Lemma \ref{prop:proD}.

 \end{mproof}

 \section{Upper bound for $p_{\mathrm{succ}}(\cN,\mathfrak{F},\rho)$ }\label{appen:thm2}

  \begin{mproof} [Proof of Theorem 4]
Since
\begin{eqnarray*}
\frac{1}{2}\min_{\cM\in \mathfrak{F}}\|\cN(\rho)-\cM(\rho)\|_{1}\leq 
\omega_{1}(\cN(\rho)),
\end{eqnarray*}
then by Theorem  2  and the definition of $p_{\mathrm{succ}}(\cN,\mathfrak{F},\rho)
$, we have 
\begin{eqnarray*}
 && p_{\mathrm{succ}}(\cN,\mathfrak{F},\rho)
-p_{\mathrm{succ}}(\cN,\mathfrak{F}, \mathcal{F}) \\
&=&\frac{1}{4}\min_{\cM\in \mathfrak{F}}\norm{\cN(\rho)-\cM(\rho)}_1-\frac{1}{2}\widetilde{\Omega}_{1}(\cN)\\
  &\leq&\frac{1}{2}(\omega_{1}(\cN(\rho))
   -\widetilde{\Omega}_{1}(\cN))\\
  & \leq& \frac{1}{2}\omega_1(\rho),
\end{eqnarray*}
where the second inequality comes from the fact that 
\begin{eqnarray*}
  \omega_{1}(\cN(\rho))-\omega_{1}(\rho)
  \leq \widetilde{\Omega}_{1}(\cN),
\end{eqnarray*}
for any $\rho\in\cD(\cH)$.
Thus,  we complete the proof.

\end{mproof}

 \section{Improvement from coherent states in channel discrimination}\label{appen:imp}
 
 \begin{mproof}[Proof of Proposition 7]
It has been shown that there exists some quantum channel $\cN_*\in MIO$ but not
$IO$, i.e, there exists some quantum state $\rho$ such that
$\cN_*(\rho)\neq \cM(\rho)$ for any $\cM\in IO$ \cite{Bu2016QIC}, which implies that 
\begin{eqnarray*}
\max_{\cM\in IO}\norm{\cN_*(\rho)-\cM(\rho)}_1>0.
\end{eqnarray*}
Thus, we have $p_{\rm succ}(\cN_*, IO, Q)>1/2$.  However, 
due to Proposition 6, we have
\begin{eqnarray*}
p_{\rm succ}(\cN_*, SIO, \cI)=p_{\rm succ}(\cN_*, IO, \cI)=p_{\rm succ}(\cN_*, MIO, \cI)=1/2,
\end{eqnarray*}
 as $\cN_*\in MIO$. Thus, we have 
\begin{eqnarray*}
p_{\rm succ}(\cN_*, IO, Q)>p_{\rm succ}(\cN_*, IO, \cI).
\end{eqnarray*}
Besides, since $SIO\subset IO$, then 
$p_{\rm succ}(\cN_*, SIO, Q)\geq p_{\rm succ}(\cN_*, IO, Q)$. 
Therefore, 
\begin{eqnarray*}
p_{\rm succ}(\cN_*, SIO, Q)> p_{\rm succ}(\cN_*, SIO, \cI).
\end{eqnarray*}

\end{mproof}

 \section{Discrimination with incoherent measuresment}\label{appen:coh_f}

  \begin{mproof} [Proof of Theorem 9]
It is easy to see that 
 \begin{eqnarray*}
&& \max_{\substack{\{\Pi, \mathbb{I}-\Pi\}\\ \Pi~ \text{diagonal}}}\left\{\frac{1}{2}\Tr{\cN(\rho)\Pi}+\frac{1}{2}\Tr{\cM(\rho)(\mathbb{I}-\Pi)}\right\}\\
&=& \max_{\substack{\{\Pi, \mathbb{I}-\Pi\}\\ \Pi~ \text{diagonal}}}\left\{\frac{1}{2}\Tr{\cN(\rho)\Delta(\Pi)}+\frac{1}{2}\Tr{\cM(\rho)(\mathbb{I}-\Delta(\Pi))}\right\}\\
&=&\max_{\substack{\{\Pi, \mathbb{I}-\Pi\}\\ \Pi~ \text{diagonal}}}\left\{\frac{1}{2}\Tr{\Delta^{\dag}\circ\cN(\rho)\Pi}+\frac{1}{2}\Tr{\Delta^{\dag}\circ\cM(\rho)(\mathbb{I}-\Pi)}\right\}\\
 &\leq & \max_{\{\Pi, \mathbb{I}-\Pi\}}\left\{\frac{1}{2}\Tr{\Delta^{\dag}\circ\cN(\rho)\Pi}+\frac{1}{2}\Tr{\Delta^{\dag}\circ\cM(\rho)(\mathbb{I}-\Pi)}\right\}\\
& =&\frac{1}{2}+\frac{1}{4}\norm{\Delta^{\dag}\circ\cN(\rho)-\Delta^{\dag}\circ\cM(\rho)}_1.
 \end{eqnarray*}
Besides, 
$\Delta^{\dag}$ satisfies the conditions
that $\Delta^{\dag}(\cD(\cH))\subset \mathcal{I}$ and $\Delta^{\dag}(\rho)=\rho$ for any $\rho\in\mathcal{I}$, which implies that
\begin{eqnarray*}
\frac{1}{2}\min_{\cM\in \mathfrak{J}}\norm{\Delta^{\dag}\circ\cN(\rho)-\Delta^{\dag}\circ\cM(\rho)}_{1}=
C_{1}(\Delta^{\dag}\circ\cN(\rho))=0.
\end{eqnarray*}

 \end{mproof}
 Note that, in any other resource theory with resource destroying channel $\lambda$,
 we  can also define the free measurement $\set{\Pi, \mathbb{I}-\Pi}$, where
 $\Pi$ and $\mathbb{I}-\Pi$ are propositional  to some free states.  Then it is to see that the above proof still works 
 for the free measurement case 
if  the resource destroying map satisfies that conditions that $\lambda^{\dag}(\cD(\cH))\subset \mathcal{F}$ and $\lambda^{\dag}(\rho)=\rho$ for any $\rho\in\mathcal{F}$, i.e, 
$\lambda^{\dag}$ is a resource destroying map \cite{Liu2017}.

\end{document}